\documentclass[letterpaper,10pt,conference]{ieeeconf}

\IEEEoverridecommandlockouts
\usepackage{cite}
\usepackage{amsmath,amssymb,amsfonts,mathtools}

\usepackage{amsthm}
\usepackage{bm}
\usepackage{mathrsfs}
\let\labelindent\relax
\usepackage{enumitem}
\usepackage{xcolor}
\usepackage{graphicx}
\usepackage{booktabs}
\usepackage{tabularx}
\usepackage{array}
\usepackage{multirow}
\usepackage{float}
\usepackage[font=small]{caption}
\usepackage{subcaption}
\usepackage[hidelinks]{hyperref}

\newtheorem{theorem}{Theorem}[section]
\newtheorem{proposition}[theorem]{Proposition}

\newtheorem{corollary}[theorem]{Corollary}
\theoremstyle{definition}
\newtheorem{definition}[theorem]{Definition}
\newtheorem{remark}[theorem]{Remark}

\newtheorem{assumption}[theorem]{Assumption}

\newcommand{\R}{\mathbb{R}}
\newcommand{\E}{\mathbb{E}}
\newcommand{\Prob}{\mathbb{P}}
\newcommand{\cE}{\mathcal{E}}
\newcommand{\cW}{\mathcal{W}}
\newcommand{\cR}{\mathcal{R}}
\newcommand{\cB}{\mathcal{B}}
\newcommand{\cA}{\mathcal{A}}

\newcommand{\cT}{\mathcal{T}}
\newcommand{\PhiMacro}{\Phi}
\newcommand{\norm}[1]{\left\lVert #1 \right\rVert}

\newcommand{\JS}{\mathrm{JS}}
\newcommand{\KL}{\mathrm{KL}}
\newcommand{\Sh}{\mathrm{Sh}}
\newcommand{\Ren}{\mathrm{Ren}}
\newcommand{\Ts}{\mathrm{Ts}}

\begin{document}

\title{Cyber Dynamics I: Finite Macrostates for Behavioral Anomaly Detection in Network Telemetry}

\author{
Abdul Rahman$^{1}$, Eranga Bandara$^{2}$, Sachin Shetty$^{2}$\\
$^{1}$arahman@alum.howard.edu \\
$^{2}$Old Dominion University, USA \\
{\rm \{cmedawer, sshetty\}@odu.edu}
}

\maketitle

\begin{abstract}
Entropy-based methods have long been used for network anomaly detection, but most existing approaches treat entropy as a scalar statistic on narrow observables rather than as part of a broader behavioral state-space for cyber systems. We propose a finite-dimensional macrostate framework for network telemetry, instantiated over the Canonical Security Telemetry Substrate (CSTS), so that coarse-graining is performed over persistent entities, typed relations, and temporal state rather than isolated event records. The resulting macrostate captures activity, distributional disorder, structural organization, temporal volatility, persistence, and deviation from benign baselines. Rather than scoring only unusual states, we model window-to-window macrostate transitions and define regime structure, stability, and anomalous change. This supports discrimination between benign workload drift and adversarial reorganization. We evaluate the framework on benchmark network telemetry datasets and compare it against Shannon-, R\'enyi-, and Tsallis-style entropy baselines, as well as standard anomaly detectors. The proposed representation improves anomaly discrimination and supports more interpretable behavioral analysis of cyber telemetry.
\end{abstract}
\section{Introduction}

\subsection{Motivation}

Modern cyber defense operates in environments that are fundamentally dynamic, heterogeneous, and behavior-rich. Enterprise networks evolve continuously under ordinary operational forces such as office-hour transitions, scheduled backups, patch cycles, software deployment bursts, batch processing, cloud autoscaling, identity-driven access changes, and topology or service reconfiguration. At the same time, adversarial activity increasingly unfolds as multi-stage behavioral campaigns rather than isolated point anomalies, especially in advanced persistent threat (APT) settings, where persistence, lateral movement, staging, and coordinated re-use of infrastructure create temporal structure that is often more informative than any single event in isolation \cite{che2024systematic,buchta2024advanced}. These realities create a central tension for anomaly detection: the system must remain sensitive to genuinely malicious change while avoiding false alarms induced by normal operational variability, concept drift, and evolving workload conditions \cite{ahmed2024driftsurvey,hinder2024driftsurveyA}.

This tension exposes a limitation of pointwise or narrowly aggregated anomaly formulations. When telemetry is reduced to isolated events or small feature vectors, the distinction between benign variation and adversarial reorganization can become blurred. A burst of backup traffic, for example, may resemble exfiltration at the level of volume alone; a patch wave may resemble scanning or coordinated service churn at the level of connection counts; and cloud elasticity can alter endpoint and service distributions in ways that superficially mimic attack-induced instability. In such environments, meaningful detection requires not only sensitivity to unusual observations, but also a representation of \emph{system behavior over time}: what entities are active, how they are related, how interaction structure shifts, and whether the system remains within a stable operational regime or transitions into an anomalous one \cite{EkleEberle2024,JinEtAl2024GNN4TS,OkolieEtAl2025HeterogeneousCyber}.

A further complication is that modern cybersecurity analytics increasingly draw from heterogeneous telemetry sources rather than from packet or flow data alone. Network events, process trees, identity activity, file operations, DNS interactions, and higher-level behavioral signals often need to be interpreted jointly. This has motivated work on schema normalization and telemetry unification, including operational standards such as the Open Cybersecurity Schema Framework (OCSF) and broader observability ecosystems such as OpenTelemetry \cite{OCSF,opentelemetry,rongali2025opentelemetry}. In our view, however, a normalized schema alone is not enough. What is also needed is a canonical \emph{behavioral substrate} on which cyber AI models can be built consistently across tasks and data modalities. This is precisely the role played by the Canonical Security Telemetry Substrate (CSTS), which elevates persistent entities, typed relations, and temporal state into first-class objects and thereby provides a natural foundation for anomaly detection, graph learning, forecasting, and behavioral security analytics \cite{RahmanCSTS}. In this paper, we leverage CSTS as the representational substrate on which macrostate construction is defined.

\subsection{From Entropy to Cyber Dynamics}

Entropy has long provided a useful language for anomaly detection. The foundational information-theoretic viewpoint originates with Shannon's formalization of uncertainty and information \cite{Shannon1948}, while Jaynes' maximum-entropy perspective showed how constrained distributions can encode macroscopic knowledge in a principled way \cite{Jaynes1957}. These ideas inspired a substantial cybersecurity literature in which entropy is computed over traffic distributions such as source addresses, destination addresses, ports, protocols, flow sizes, or related categorical observables \cite{LeeXiang2001,LakhinaCrovellaDiot2005,GuMcCallum2005,BerezinskiJasiul2015}. In practice, Shannon entropy and its generalized relatives, including R\'enyi and Tsallis formulations, have proven useful for identifying scans, worms, denial-of-service behavior, botnet activity, and other deviations in traffic composition \cite{GuMcCallum2005,BerezinskiJasiul2015}. This literature is important and should not be minimized: it established that disorder-like statistics can reveal meaningful departures from nominal network behavior.

At the same time, most entropy-based cyber methods remain narrow in two senses. First, entropy is often applied to a small family of marginal histograms rather than to a richer representation of system organization. Second, entropy is typically used as a scalar feature or thresholding statistic, not as one coordinate within a higher-dimensional behavioral state-space. As a result, these methods frequently detect ``unusual distributions'' without explicitly representing whether the underlying system is becoming more coordinated, more concentrated, more volatile, more persistent, or more structurally reorganized. Yet many attacks are not simply episodes of high randomness or large volume. They are better understood as \emph{organized reconfiguration}: an attacker induces new communication motifs, establishes persistent command pathways, concentrates privilege flows, stages lateral transitions, or perturbs interaction structure in ways that may preserve some marginal distributions while fundamentally changing the system's behavioral organization \cite{che2024systematic,buchta2024advanced,EkleEberle2024}.

This observation motivates a shift from scalar entropy toward \emph{cyber dynamics}. By cyber dynamics, we mean a finite-dimensional, coarse-grained description of a cyber system in which raw micro-events within a time window are mapped to interpretable macro-observables and then studied as a temporal trajectory. In this view, entropy is retained as an important component, but it is not the endpoint. Instead, it becomes one member of a broader macrostate vector that may also encode activity, structural order, volatility, persistence, coupling, and deviation from nominal regimes. The central analytical object is no longer a single entropy score, but the evolution of a macrostate through time, together with the stability of benign regimes and the transition structure connecting them \cite{Shannon1948,Jaynes1957,EkleEberle2024,JinEtAl2024GNN4TS}. This shift is especially natural when telemetry is represented over CSTS, since persistent entities, typed relations, and temporal state provide exactly the ingredients needed for coarse-grained behavioral modeling \cite{RahmanCSTS}.

Our starting thesis is therefore the following: cyber systems admit coarse-grained finite macrostates whose temporal dynamics encode normal regimes, benign drift, perturbation, and adversarial change. The immediate goal of the present paper is not to claim a fully complete cyber thermodynamics, nor to assert that physics analogies should be imported uncritically. Rather, the goal is to establish a first rigorous step: a finite macrostate formalism for windowed network telemetry, instantiated over CSTS, that moves beyond entropy-only detection toward interpretable behavioral-state analysis. In this sense, entropy serves as inspiration, but the object of study is dynamics.

\subsection{Main Contributions}

The contributions of this paper are as follows:

\begin{itemize}[leftmargin=2em]
    \item We introduce a \emph{finite macrostate formalism} for windowed network telemetry in which raw micro-events are coarse-grained into interpretable macro-observables over activity, disorder, structure, volatility, persistence, and baseline deviation.
    
    \item We instantiate this framework over the \emph{Canonical Security Telemetry Substrate (CSTS)}, thereby grounding macrostate construction in persistent entities, typed relations, and temporal state rather than isolated event records \cite{RahmanCSTS}.
    
    \item We move beyond scalar entropy by defining a \emph{behavioral state-space} for cyber systems and by treating anomaly detection as both a \emph{state} problem and a \emph{transition} problem, allowing benign workload drift to be distinguished from adversarial reorganization.
    
    \item We provide a finite macrostate formalism for windowed network telemetry, instantiated over CSTS, that moves beyond entropy-only detection toward interpretable behavioral-state analysis.
    
    \item We evaluate the proposed representation on benchmark network telemetry datasets and compare it against Shannon-, R\'enyi-, and Tsallis-style entropy baselines, as well as standard anomaly detectors, showing improved anomaly discrimination and showing improved anomaly discrimination and more interpretable behavioral analysis of cyber telemetry. \cite{LeeXiang2001,LakhinaCrovellaDiot2005,GuMcCallum2005,BerezinskiJasiul2015}.
    
    \item More broadly, we position this paper as the first step in a larger \emph{cyber dynamics} program whose later stages extend naturally to heterogeneous telemetry, temporal graphs, and behavior-centric cyber AI models \cite{OkolieEtAl2025HeterogeneousCyber,EkleEberle2024,JinEtAl2024GNN4TS}.
\end{itemize}

\subsection{Paper Organization}

The remainder of the paper is organized as follows. Section~II reviews prior work on entropy-based anomaly detection, generalized entropy and divergence methods, temporal and graph anomaly analysis, and heterogeneous cybersecurity telemetry. Section~III introduces the formal framework of microstates, windowed ensembles, and finite macrostates, and explains how the construction is instantiated over CSTS. Section~IV defines the macro-observable families used in the paper, including activity, disorder, structural order, volatility, persistence, and deviation variables. Section~V develops the dynamical viewpoint through macrostate trajectories, nominal transition operators, regime structure, and stability notions. Section~VI presents the basic theory, including the proposition and theorem stack that motivates coarse-graining, transition-aware anomaly scoring, and multiscale analysis. Section~VII details the experimental design and benchmark setup, and Section~VIII presents the empirical results. Section~IX discusses implications, limitations, and the path toward broader cyber dynamics papers over heterogeneous telemetry and behavioral cyber AI.

\section{Related Work}

\subsection{Entropy-Based Network Anomaly Detection}

Entropy-based anomaly detection has a long and legitimate history in cybersecurity and network measurement. At the foundational level, Shannon's formulation of information and uncertainty established entropy as a natural measure of distributional dispersion \cite{Shannon1948}, while Jaynes' maximum-entropy program provided a principled way to model macroscopic regularity from limited constraints \cite{Jaynes1957}. These ideas entered network anomaly detection through information-theoretic formulations that treated traffic attributes such as source addresses, destination addresses, ports, protocols, and related categorical distributions as observables whose entropy profiles could reveal unusual activity \cite{LeeXiang2001}. This line of work was strengthened by distribution-based traffic mining approaches showing that anomalies often manifest not merely as total-volume changes but as shifts in the empirical distributions of key traffic features \cite{LakhinaCrovellaDiot2005}.

A particularly influential step was the use of maximum-entropy estimation to characterize nominal traffic behavior and detect departures from it \cite{GuMcCallum2005}. This formulation helped move the field beyond naive thresholding on raw counts by emphasizing constrained distributional modeling. Subsequent work further established entropy as a practical anomaly signal across diverse network settings, including scans, botnets, denial-of-service behavior, and other departures in flow composition \cite{BerezinskiJasiul2015}. The practical appeal of this literature is clear: entropy can be computed efficiently, it often has interpretable operational meaning, and it can detect irregularity even when signature-based systems fail.

At the same time, classical entropy-based network methods usually remain tied to a limited set of observables and are often deployed as feature-level heuristics or scalar alarms. Typical constructions compute entropy over one or several marginals---for example, over source IPs, destination IPs, ports, protocols, or packet-size bins---and then compare those values to historical baselines or threshold rules \cite{LeeXiang2001,LakhinaCrovellaDiot2005,GuMcCallum2005,BerezinskiJasiul2015}. This is effective for some forms of distributional change, but it leaves open a larger representational question: how should one describe the behavior of an evolving cyber system when the relevant signal lies not only in marginal disorder but also in structural organization, persistence, coordination, and temporal transition? That question motivates the present work.

\subsection{Generalized Entropy and Divergence Approaches}

The classical Shannon formulation has also been generalized in several directions. R\'enyi and Tsallis entropy families provide tunable alternatives that emphasize different parts of a distribution, making them attractive for network security settings where rare events, heavy concentration, or tail sensitivity matter \cite{BerezinskiJasiul2015}. Likewise, relative-entropy and divergence-based formulations such as Kullback--Leibler and Jensen--Shannon comparisons support the direct measurement of drift between current traffic distributions and nominal baselines \cite{GuMcCallum2005,BerezinskiJasiul2015}. From an operational perspective, these methods allow one to compare not just absolute entropy levels, but the \emph{distance} between present and expected behavior.

This generalized entropy and divergence literature is important because it shows that the anomaly signal is often richer than a single Shannon scalar. Different entropy families emphasize different distributional regimes, and divergence measures can capture changes even when absolute entropy itself is ambiguous. Nevertheless, these approaches usually preserve the same basic architecture as the classical entropy methods: they act on selected feature distributions and then feed the resulting quantities into thresholding, ranking, or downstream detectors. In other words, even when the entropy functional changes, the representational object is still usually a small set of feature-level summaries rather than a full behavioral state-space.

This distinction matters for the present paper. Our claim is not that generalized entropy is unimportant; on the contrary, it remains one of the most useful inspiration sources for cyber anomaly detection. Rather, the limitation is that Shannon, R\'enyi, Tsallis, KL, and JS style quantities are usually treated as \emph{end products} of analysis rather than as \emph{coordinates} inside a larger macrostate representation. A system may preserve several traffic marginals while nonetheless undergoing substantial reorganization in who communicates with whom, how persistent certain pathways become, or how interaction motifs evolve over time. Such behavior may not be captured fully by even a well-chosen family of entropy or divergence statistics alone \cite{GuMcCallum2005,BerezinskiJasiul2015}. This is one of the main reasons we treat entropy as an ingredient in a broader macro-dynamical construction rather than as the sole analytical object.

\subsection{Temporal, Graph, and Structural Anomaly Detection}

A second major body of work has shifted attention from static feature summaries toward temporal and structural models. This includes concept-drift and distribution-shift literature emphasizing that the meaning of ``normal'' changes over time, often for reasons unrelated to adversarial activity \cite{ahmed2024driftsurvey,hinder2024driftsurveyA}. In practical cyber settings, workloads evolve under changing user populations, deployment cycles, policy updates, and environmental shifts, so anomaly detectors must separate benign nonstationarity from genuine threat behavior. This literature is directly relevant to our framing because it reinforces the idea that the correct object of study is not merely an anomalous point, but an evolving process.

Related developments in dynamic and temporal graph anomaly detection further strengthen this shift. Recent surveys show that dynamic interaction graphs, temporal graph neural models, motif-based methods, and structural anomaly formulations are now central tools for representing evolving systems whose anomalies are partly relational and partly temporal \cite{EkleEberle2024,JinEtAl2024GNN4TS}. These methods are highly relevant to cybersecurity because many attacks are relational by nature: scans alter fan-out structure, lateral movement perturbs host-user-process connectivity, beaconing induces repeated low-volume communication motifs, and staged attacks create persistent or coordinated paths that may be invisible in simple scalar summaries. Graph-based and temporal methods therefore capture an important truth: anomaly often lives in \emph{organization} and \emph{transition}, not only in instantaneous magnitude.

Our work is aligned with this structural turn, but differs in emphasis. Much graph-based anomaly detection is representation-heavy and model-specific, often depending on embeddings, deep architectures, or graph-neural formulations tailored to particular data types \cite{EkleEberle2024,JinEtAl2024GNN4TS}. These are powerful, but they do not by themselves provide an interpretable finite macrostate formalism. The present paper instead seeks a middle layer: a compact, interpretable, coarse-grained state-space that can absorb entropy-like, structural, and temporal information while remaining usable across model families. In this sense, graph and temporal anomaly work support our structural-order lane, but our objective is to formulate that lane in a more explicit macro-dynamical language.

\subsection{Behavioral and Heterogeneous Cyber Telemetry}

A third relevant strand concerns behavioral cybersecurity and heterogeneous telemetry fusion. Modern attacks, especially advanced persistent threats, are inherently multi-stage and behavioral: they involve combinations of access establishment, reconnaissance, privilege movement, persistence, evasion, and command or exfiltration behavior unfolding over time \cite{che2024systematic,buchta2024advanced}. This has driven growing interest in systems that reason jointly over network, endpoint, identity, and higher-level behavioral signals rather than relying on packet or flow data in isolation. Recent work on anomaly detection in heterogeneous cybersecurity data reflects this need by emphasizing that realistic detection increasingly requires unified treatment of multiple modalities, schemas, and event types \cite{OkolieEtAl2025HeterogeneousCyber}.

Operational frameworks such as OCSF and observability standards such as OpenTelemetry represent important steps toward normalizing telemetry across sources \cite{OCSF,opentelemetry,rongali2025opentelemetry}. However, normalization alone does not yet provide a canonical behavioral substrate for cyber AI. That broader role is more naturally served by CSTS, which was designed as a canonical security telemetry substrate centered on persistent entities, typed relations, provenance, and temporal validity \cite{RahmanCSTS}. In the context of the present paper, this matters because a macrostate framework should not be built over isolated records if it is intended to generalize beyond a single dataset or detection task. By instantiating our windowed ensembles over CSTS, we retain a path from network-focused anomaly detection in this paper to future work on endpoint, identity, graph, and fully heterogeneous cyber dynamics without changing the underlying representational philosophy.

We emphasize, however, this paper remains network-focused. The present section therefore situates heterogeneous telemetry not as a fully developed empirical contribution of this paper, but as an important continuation path. The current work establishes the macrostate and dynamics formalism in the setting of network telemetry first, while deliberately choosing a substrate that supports extension to broader behavioral cyber AI.

\subsection{Gap and Positioning}

The literature reviewed above is substantial and valuable. Classical entropy methods established that network anomalies can often be detected through changes in distributional disorder \cite{LeeXiang2001,LakhinaCrovellaDiot2005,GuMcCallum2005,BerezinskiJasiul2015}. Generalized entropy and divergence methods broadened that toolkit, allowing richer sensitivity to concentration, rarity, and baseline drift \cite{GuMcCallum2005,BerezinskiJasiul2015}. Temporal, graph, and structural anomaly work showed that evolving interaction structure and time dependence are often essential \cite{EkleEberle2024,JinEtAl2024GNN4TS}. Behavioral and heterogeneous telemetry research made clear that modern cyber defense increasingly depends on multi-source, behavior-centered analysis \cite{che2024systematic,buchta2024advanced,OkolieEtAl2025HeterogeneousCyber}. Yet despite these advances, there is still no broadly adopted finite macrostate framework that simultaneously:
\begin{enumerate}[label=(\roman*),leftmargin=2em]
    \item coarse-grains temporal network telemetry into a compact set of interpretable macro-variables,
    \item treats anomaly as both a \emph{state} and a \emph{transition} phenomenon in macrostate space,
    \item integrates entropy, structure, persistence, and volatility within a single behavioral representation, and
    \item is naturally instantiated over a canonical telemetry substrate that supports extension to broader cyber AI tasks.
\end{enumerate}

This is the gap addressed by the present paper. Our novelty is therefore not the claim that entropy is new, nor that temporal or graph anomaly detection is new. Rather, the novelty lies in introducing a \emph{finite macro-dynamical formalism} for network telemetry, instantiated over CSTS, in which entropy becomes one coordinate of a broader behavioral state-space and anomalies are studied through macrostate regimes and regime transitions. This positions the paper as a bridge between entropy-based anomaly detection, structural cyber analytics, and a more general program of cyber dynamics.

\section{Framework: Microstates, Windows, and Macrostates}

This section introduces the formal language used throughout the paper and, more broadly, throughout the cyber dynamics program. Our goal is to define a representation layer that is general enough to support multiple telemetry types and model families, while remaining concrete enough for immediate use in network anomaly detection. The guiding principle is that cyber systems should not be modeled solely as sequences of isolated records, but rather as evolving collections of entities, relations, and state changes that admit meaningful coarse-graining over time. In this paper, we instantiate that perspective in the network setting and do so in a manner consistent with the Canonical Security Telemetry Substrate (CSTS), which provides a natural entity-relational and temporally indexed substrate for cyber AI \cite{RahmanCSTS}. The resulting formalism begins at the level of event microstates, aggregates them into windowed micro-ensembles, and then maps those ensembles into finite-dimensional macrostates whose trajectories can be studied dynamically.

\subsection{Event-Level Microstates}

We begin with the smallest observational unit used in the present framework. For this work such a unit is a typed network event, but the construction is intentionally general so that the same formal language can later be reused for endpoint telemetry, identity telemetry, provenance events, or other heterogeneous security signals. The central requirement is that the primitive object preserve time, participating entities, interaction type, and event attributes in a way that is compatible with canonical telemetry modeling.

\begin{definition}[Network microstate]
A network microstate is a typed event
\[
e=(\tau, \mathrm{src}, \mathrm{dst}, \mathrm{svc}, \mathrm{proto}, \mathrm{dir}, \mathrm{attrs}),
\]
where $\tau$ is a timestamp, $\mathrm{src}$ and $\mathrm{dst}$ are endpoint or role identifiers,
$\mathrm{svc}$ is a service or port class, $\mathrm{proto}$ is the protocol, $\mathrm{dir}$ denotes directionality,
and $\mathrm{attrs}$ collects event-specific numerical or categorical attributes such as bytes, packets, duration,
flags, or derived role labels.
\end{definition}

This definition is deliberately minimal. It captures the fact that a network event is not simply a number or a label, but a temporally localized interaction between typed entities under a protocol and service context. The tuple form also allows the modeler to choose an appropriate granularity for the source and destination identifiers. In some settings these may be raw IP addresses; in others they may be host identifiers, subnet roles, asset classes, internal/external designations, or other abstractions derived from context. Likewise, the service field may correspond to raw ports, port classes, application protocols, or semantically grouped service categories. The framework does not require that the finest-grained identifiers be retained, only that the chosen representation preserve enough information to define meaningful empirical distributions, interaction structures, and behavioral aggregates.

\begin{remark}
The representation above is intentionally generic and can be instantiated for NetFlow, Zeek logs, packet-derived summaries, or future endpoint and heterogeneous telemetry. In the broader CSTS perspective, the tuple $e$ should be interpreted as a time-indexed observation over persistent entities and typed relations rather than as an isolated flat record \cite{RahmanCSTS}. This allows later papers to replace network-only identifiers with richer entity and relation types without changing the surrounding formalism.
\end{remark}

A key design choice is that the framework does not bind itself to one vendor schema, one dataset format, or one telemetry source. Instead, the microstate is the abstract object that remains after normalization and canonicalization. This is important for two reasons. First, it makes the macrostate construction portable across benchmarks and deployment settings. Second, it aligns the paper with a broader view of cyber AI in which learning and inference should be built over canonical entity-relational substrates rather than over one-off data layouts. In the present paper, however, we remain deliberately conservative: the empirical evaluation focuses on network telemetry, and the richer generality of the microstate definition serves primarily to keep the framework reusable.

\subsection{Windowed Micro-Ensembles}

Individual events are often too fine-grained to support robust behavioral interpretation. Operationally meaningful cyber patterns usually emerge only after aggregation over short temporal horizons: a scan is visible as repeated contact behavior across destinations, a backup burst appears as sustained high-volume service activity, and beaconing emerges through repeated low-volume interactions with characteristic regularity. For this reason, the second level of the framework is not the isolated event but the \emph{windowed ensemble} of events collected over a bounded time interval.

\begin{definition}[Time window]
Fix a window length $\Delta>0$ and stride $\delta>0$. A window is an interval
\[
W_t=[t,t+\Delta).
\]
The window family is
\[
\cW=\{W_t : t = k\delta,\; k\in \mathbb{Z}_{\ge 0}\}.
\]
\end{definition}

The parameters $\Delta$ and $\delta$ control both the temporal resolution and the sampling fidelity of the analysis.  The window length determines how much local behavioral context is retained, while the stride determines how frequently that context is observed.  When $\delta<\Delta$, successive windows overlap.  This overlap is not merely a smoothing device.  It also improves boundary sensitivity: a behavioral shift may occur across the
boundary between two disjoint windows and therefore fail to appear as a coherent pattern in either window separately.  Overlapping windows increase the probability that such boundary-spanning behavior is captured inside at least one local ensemble, making short-lived transitions, onset behavior, and regime changes more visible in the induced macrostate trajectory.

By contrast, non-overlapping windows with $\delta=\Delta$ provide a simpler partition of the event stream and reduce computational redundancy, but they can be less sensitive to transient or boundary-localized changes.  For this reason, overlapping windows are especially natural in a transition-aware macrostate framework: they preserve a denser sampling of the path $\{X_t\}_{t\in T}$ and support more faithful estimation of state change, regime movement, and anomalous reorganization across time.

\begin{definition}[Windowed micro-ensemble]
The micro-ensemble induced by $W_t$ is
\[
\cE_t := \{e : \tau(e)\in W_t\}.
\]
\end{definition}

The ensemble notation emphasizes that the object of interest is not merely a set of timestamps, but a finite multiset of typed events whose aggregate behavior can be summarized statistically and structurally. Depending on the application, the ensemble may be viewed in several equivalent ways: as a multiset of records, as a temporal slice of a typed interaction graph, as an empirical sample from a latent behavioral regime, or as a local realization of a cyber process over a short horizon. All of these viewpoints are compatible with the present framework. What matters is that $\cE_t$ contains enough local evidence to support the construction of macro-observables describing load, disorder, structure, persistence, and deviation from nominal behavior.

This windowed-ensemble perspective already departs from traditional pointwise anomaly detection. Instead of asking whether a single event is unusual, we ask what kind of local behavioral configuration the system currently occupies. This distinction is central to the rest of the paper. Many anomalous phenomena are not encoded in isolated observations, but in the coordinated pattern formed by many micro-events over time. The ensemble is the first level at which that pattern becomes visible.

\subsection{Coarse-Graining Map}

The central formal move of the paper is a coarse-graining step that maps each micro-ensemble into a finite-dimensional macrostate. This is the point at which the language of entropy and dynamics becomes operational. Rather than reducing each window to a single disorder score, we extract a vector of interpretable macro-observables, each intended to capture a different aspect of system behavior. The macrostate is therefore a compressed but structured description of the local cyber configuration induced by the window.

\begin{definition}[Macrostate map]
A macrostate map is a function
\[
\PhiMacro : \{\text{finite event multisets}\} \to \R^d
\]
that assigns to each micro-ensemble $\cE_t$ a finite-dimensional macrostate
\[
X_t := \PhiMacro(\cE_t).
\]
\end{definition}

The purpose of $\PhiMacro$ is not simply dimensionality reduction in the machine-learning sense. It is a principled coarse-graining operator. It translates local collections of micro-events into a behavioral state vector that can be analyzed geometrically and dynamically. In particular, $\PhiMacro$ should preserve enough information to distinguish operational regimes, identify unusual transitions, and separate benign drift from adversarial reorganization. At the same time, it should be sufficiently compact and interpretable that the resulting macrostate space can be reasoned about directly.

\begin{definition}[Cyber macrostate]
The macrostate at time window $W_t$ is the vector
\[
X_t=(A_t,H_t,S_t,V_t,P_t,C_t,D_t)\in \R^d,
\]
where:
\begin{itemize}[leftmargin=2em]
    \item $A_t$ denotes activity/load variables,
    \item $H_t$ distributional disorder variables,
    \item $S_t$ structural order variables,
    \item $V_t$ temporal volatility variables,
    \item $P_t$ persistence/memory variables,
    \item $C_t$ coupling/influence variables,
    \item $D_t$ deviation-to-baseline variables.
\end{itemize}
\end{definition}

Each block of coordinates plays a distinct role. The activity variables summarize overall system occupancy and workload intensity. The disorder variables capture entropy-like dispersion across relevant empirical distributions. The structural variables encode organization in the interaction pattern, such as concentration, connectivity, or motif structure. The volatility variables describe short-horizon instability or burstiness. The persistence variables encode repeated or durable relational patterns. The coupling variables measure dependencies or influence between behavioral components. Finally, the deviation variables compare the current window to benign baselines or nominal regime summaries. Together, these coordinates define a finite behavioral state-space rich enough to support more than thresholding on one statistic, yet compact enough to remain interpretable.

The decomposition of $X_t$ should not be misunderstood as asserting that these seven blocks are uniquely determined or universally optimal. Rather, they define a canonical \emph{macro-basis} for cyber dynamics: a finite family of observable types that can be instantiated differently across domains while preserving the same formal architecture. For network telemetry, the empirical realization may involve counts, entropies, graph summaries, burstiness indicators, recurrence ratios, and baseline divergences. For endpoint or heterogeneous telemetry, the same block structure can be populated with process, identity, provenance, or multi-modal relational statistics. This portability is one of the main reasons to build the framework over CSTS, where entities, relations, and temporal validity are already explicit \cite{RahmanCSTS}.

\subsection{Behavioral Regimes and Anomalies}

Once macrostates have been defined, the next step is to view cyber behavior as motion through macrostate space. This leads naturally to the concepts of regimes, benign sets, and anomalous transitions. The language is intentionally dynamical: a system occupies local behavioral configurations, revisits some of them regularly, drifts among them under benign operational forces, and may leave them under attack or abnormal perturbation.

\begin{definition}[Behavioral regime]
A behavioral regime is a subset $\cR\subset \R^d$ of macrostate space representing a recurrent operational mode.
\end{definition}

A regime should be understood as a coarse operational state of the system rather than as a single point. Examples include normal office-hour activity, overnight quiet states, backup bursts, patching windows, scanning behavior, staging behavior, or other recurrent configurations that occupy coherent regions of macrostate space. Importantly, a regime may have internal variability. What makes it a regime is not that all windows are identical, but that they belong to a recurrent and behaviorally interpretable part of the state-space.

\begin{definition}[Benign regime set]
The benign regime set $\cB\subset \R^d$ is the union of recurrent macrostates corresponding to nominal operations.
\end{definition}

The benign set plays the role of a reference object for detection. It need not be a single cluster, nor even a convex region. Real systems often exhibit multiple nominal modes---for example, daytime traffic, nighttime batch processing, maintenance cycles, or other scheduled behaviors. A useful anomaly framework must therefore admit the possibility that normality is multimodal and dynamic.

\begin{definition}[State anomaly]
A window $W_t$ is state-anomalous if $X_t$ lies sufficiently far from $\cB$ under a chosen macrostate distance.
\end{definition}

This is the natural generalization of classical anomaly scoring to macrostate space. Rather than comparing a feature vector or entropy scalar to a threshold, we compare the current behavioral state to the set of benign regimes. In practical terms, this allows multiple forms of nominal variation to coexist while still identifying windows that fall outside the operational envelope.

\begin{definition}[Transition anomaly]
A transition $X_t\mapsto X_{t+1}$ is anomalous if it deviates sufficiently from the nominal transition law on benign windows.
\end{definition}

The notion of transition anomaly is one of the paper's central conceptual moves. A window may appear individually benign while the transition into it is highly unusual, or conversely a window may have moderate state deviation but arise through a strongly suspicious trajectory. This matters operationally because many cyber threats are better characterized by how the system moves than by where it happens to be at one instant. Benign operational drift, for example, may produce large but smooth state changes, while coordinated attack behavior may create abrupt or structurally inconsistent transitions even when some marginals remain near normal.

\begin{remark}
The present framework studies both \emph{state anomalies} and \emph{transition anomalies}. This dual viewpoint is one of the main conceptual differentiators of the paper. Classical entropy-based methods typically focus on unusual distributions at a given time, whereas the macro-dynamical view developed here treats anomaly as a property of both the current behavioral configuration and the path by which the system arrived there.
\end{remark}

Taken together, the definitions in this section establish the formal language needed for the remainder of the paper. We have defined event-level microstates, aggregated them into windowed ensembles, mapped those ensembles into finite macrostates, and introduced the regime-based notions through which benign behavior, state anomaly, and transition anomaly will be analyzed. The next section specifies the macro-observable families in more concrete detail, including the activity, disorder, structure, volatility, persistence, coupling, and deviation variables that populate the macrostate vector.

\section{Macro-Observable Design}

The macrostate map introduced in the previous section is only useful if its coordinates are chosen carefully. They must be rich enough to capture meaningful behavioral variation, structured enough to support interpretation, and practical enough to be instantiated on benchmark network telemetry. In this paper we therefore adopt a finite macro-basis organized into six main observable families together with a deviation layer relative to benign baselines. The design is motivated by the central claim of the paper: a cyber system should not be summarized by disorder alone, but by a joint description of activity, disorder, structure, volatility, persistence, coupling, and deviation. This section defines those observable families at the level needed for both formal analysis and experimental implementation.

A guiding principle throughout is that the coordinates should be interpretable under coarse-graining. They are not arbitrary engineered features. Rather, they are intended to function as \emph{macro-observables}: finite-dimensional summaries that preserve behaviorally meaningful properties of a windowed ensemble while remaining compatible with a CSTS-style entity-relational representation \cite{RahmanCSTS}. For this paper the empirical focus is on network telemetry, but the same macro-basis is designed to transfer to richer telemetry settings in later work.

\subsection{Activity and Occupancy Variables}

The first coordinate family captures the overall scale and occupancy of activity within a window. Let
\[
A_t=(a_1(t),\dots,a_{d_A}(t)).
\]
These variables describe how much activity is occurring and how widely that activity is distributed across participating entities, services, and communication roles. In the network setting, the most natural examples include the total number of events or flows, total bytes, total packets, unique source count, unique destination count, number of active services, number of active protocols, and role-based occupancy measures such as the number of internal hosts, external destinations, or active subnet classes observed within the window.

Activity variables are indispensable for at least two reasons. First, they encode the coarse scale of the system, which is often necessary to interpret all other observables. A change in entropy or motif concentration may have a very different meaning in a sparse overnight regime than in a high-load daytime regime. Second, they provide the basic context needed to distinguish low-volume structured behavior from high-volume but otherwise routine operational activity. In practice, this paper will instantiate this family using simple, robust quantities such as total flows, total bytes, unique source roles, unique destination roles, service occupancy, active subnet counts, and selected duration summaries such as median or upper quantile flow duration.

These variables are not themselves anomaly detectors. A backup window, a denial-of-service burst, and a large deployment wave may all register as high activity. Their value lies instead in anchoring the macrostate: they define how large the local system configuration is before one asks whether it is disordered, concentrated, volatile, or behaviorally unusual.

\subsection{Distributional Disorder Variables}

The second observable family captures disorder-like properties of the empirical distributions induced by a windowed ensemble. This is the natural place where the classical entropy literature enters the macrostate. Let $p_t^{(m)}$ denote an empirical categorical distribution induced by a chosen observable family $m$, for example source roles, destination roles, service classes, protocol classes, flow-size bins, duration bins, or other discrete partitions of the event ensemble. Define
\[
H_t^{\Sh,m} = -\sum_i p_t^{(m)}(i)\log p_t^{(m)}(i).
\]
For $\alpha>0$, $\alpha\neq 1$, define the R\'enyi entropy
\[
H_t^{\Ren,m}(\alpha)=\frac{1}{1-\alpha}\log\left(\sum_i p_t^{(m)}(i)^\alpha\right),
\]
and for $q>0$, $q\neq 1$, define the Tsallis entropy
\[
H_t^{\Ts,m}(q)=\frac{1-\sum_i p_t^{(m)}(i)^q}{q-1}.
\]

This family is directly motivated by the classical and generalized entropy literature in network anomaly detection \cite{Shannon1948,Jaynes1957,LeeXiang2001,LakhinaCrovellaDiot2005,GuMcCallum2005,BerezinskiJasiul2015}. The essential role of disorder variables is to quantify how concentrated or dispersed activity is across relevant categories. A scan may flatten a destination distribution while increasing fan-out concentration in another sense; beaconing may preserve low volume but induce repeated concentration on a service or destination class; exfiltration may alter byte-bin or service-use distributions; and ordinary workload shifts may redistribute traffic over roles, protocols, or subnets without implying malicious intent.

In the macrostate perspective, however, these entropy quantities are not treated as endpoints. They are coordinates. We therefore write
\[
H_t=\big(H_t^{(1)},\dots,H_t^{(d_H)}\big),
\]
where each component may be a Shannon, R\'enyi, or Tsallis quantity evaluated on a specific observable family. In practice, it is neither necessary nor desirable to include every possible entropy variant. This paper uses Shannon entropy as the primary anchor and includes only a small number of generalized-entropy coordinates where they add meaningful sensitivity to concentration or tail structure. The broader point is that disorder enters the macrostate as one family among several, not as the sole summary of behavior.

\subsection{Structural Order Variables}

Entropy alone cannot capture how a system is organized. Two windows may have identical source, destination, or service distributions while differing radically in communication topology, concentration of influence, or coordination pattern. To capture this aspect of behavior we introduce structural-order variables. For each window $W_t$, let $G_t$ be the typed interaction graph induced by $\cE_t$, where vertices represent participating entities or role classes and edges encode observed interactions, optionally typed by service, protocol, or direction.

Possible coordinates of $S_t$ include:
\begin{itemize}[leftmargin=2em]
    \item degree concentration,
    \item graph density,
    \item motif counts,
    \item edge-type diversity,
    \item bipartite service-access concentration,
    \item role-transition concentration.
\end{itemize}

These quantities are motivated by the dynamic-graph and structural anomaly literature, which has shown that evolving interaction structure often contains anomaly signals not visible in marginal feature statistics \cite{EkleEberle2024,JinEtAl2024GNN4TS}. In the present paper, structural order refers not to order in a thermodynamic sense, but to coarse relational organization: how concentrated the communication graph is, how many repeated local patterns it contains, whether activity is spread diffusely or focused through a small number of hubs, and whether interactions remain behaviorally balanced across roles and service-access channels.

For network telemetry, useful structural summaries include measures of source or destination degree concentration, density of the induced host-service graph, ratios of one-to-many or many-to-one motifs, concentration of access into a small number of service classes, and role-transition concentration between internal, external, client, server, and other role partitions. These quantities are especially important for the paper's claim that attacks are often better described as \emph{organized reconfiguration} rather than as mere increases in randomness. A window may exhibit moderate entropy change while simultaneously showing strong centralization, repeated motif formation, or other structural regularities indicative of malicious coordination.

\subsection{Temporal Volatility Variables}

A macrostate should also reflect how behavior fluctuates locally in time. Static summaries of an entire window can miss short-lived bursts, irregular interarrival behavior, oscillations, or abrupt micro-transitions that are behaviorally important. We therefore include a temporal-volatility family,
\[
V_t=(v_1(t),\dots,v_{d_V}(t)),
\]
whose coordinates capture local instability or variation either within the window itself or across a short recent sequence of windows.

Possible coordinates of $V_t$ include:
\begin{itemize}[leftmargin=2em]
    \item sub-window burstiness,
    \item within-window interarrival irregularity,
    \item short-lag variation of selected counts,
    \item transition entropy on discretized event subtypes,
    \item change-point proxy statistics.
\end{itemize}

Operationally, volatility variables are useful because benign regimes are rarely static, yet their fluctuations often have characteristic forms. Daytime traffic may be high but relatively smooth; patching may produce sharp but expected transient bursts; beaconing may manifest as low-volume but highly regular repetition; and active attack phases may create abrupt directional changes or oscillatory patterns. The concept-drift and nonstationarity literature underscores that one must distinguish normal temporal adaptation from genuinely suspicious instability \cite{ahmed2024driftsurvey,hinder2024driftsurveyA}. Volatility coordinates support precisely that separation by enriching the macrostate with short-horizon dynamical information.

In practice, these variables may be computed by subdividing a window into shorter bins and summarizing the fluctuations of counts, byte totals, service usage, or edge activations across those bins. They may also incorporate short-lag variations across neighboring windows when a smoother trajectory-based description is desired. The exact implementation is discussed in the experimental section, but the conceptual role is already clear: volatility variables measure how a regime moves internally, not just where it sits.

\subsection{Persistence and Memory Variables}

The next family captures the extent to which behavior repeats or persists. Many cyber phenomena are not characterized only by what appears within a single window, but by whether specific entities, paths, services, or interaction motifs remain active across time. Persistence is therefore a crucial bridge from local anomaly scoring to regime analysis.

Possible coordinates of $P_t$ include:
\begin{itemize}[leftmargin=2em]
    \item repeated edge persistence,
    \item repeated service-path persistence,
    \item dwell or recurrence statistics,
    \item autocorrelation-derived summaries across recent windows.
\end{itemize}

Persistence variables encode whether relationships that appear in one window tend to recur in subsequent windows, whether service-access paths remain stable over time, and whether a regime exhibits memory. In the network setting, examples include the fraction of source-destination pairs repeated from recent windows, recurrence of source-service or destination-service paths, repeated low-volume communications suggestive of beaconing, and temporal stability of selected role-transition channels. These observables are important because benign infrastructure often shows stable recurring pathways, but adversarial persistence may also create durable coordination channels or repeated staging behavior. The distinction lies not in persistence alone, but in how persistence interacts with the other macro-observables.

This paper uses this family primarily to enrich the macrostate and to prepare the way for the more explicitly dynamical analysis of stability and recovery developed in later papers. Even at this first stage, however, persistence variables help distinguish transient bursts from sustained reorganization and support the claim that behavior should be understood through evolving regimes rather than isolated snapshots.

\subsection{Deviation-to-Baseline Variables}

The final family provides direct measures of deviation from nominal behavior. While the other macro-observables describe the internal structure of a windowed ensemble, deviation variables explicitly compare the present window to benign reference behavior. Let $p_t$ denote a current empirical distribution and let $p_{\mathrm{benign}}$ denote the corresponding benign baseline distribution. Then representative coordinates include
\[
D_t^{\JS} = \JS(p_t \,\|\, p_{\mathrm{benign}}),
\qquad
D_t^{\KL} = \KL(p_t \,\|\, p_{\mathrm{benign}}),
\]
together with geometric distances to the benign macrostate region.

Deviation variables are operationally important because anomaly detection is ultimately comparative. A macrostate may have high entropy, low entropy, high density, or strong persistence, but none of these is anomalous in itself unless interpreted relative to the system's nominal regimes. Divergence-based quantities therefore provide a direct bridge between the macrostate framework and deployable scoring methods. They also connect naturally to prior maximum-entropy and relative-entropy approaches in network anomaly detection \cite{GuMcCallum2005,BerezinskiJasiul2015}. In the present framework, however, they are generalized from feature-level distribution comparisons to regime-aware macrostate comparisons. This means that deviation may be defined not only against a benign histogram, but against a benign manifold, cluster, or region in the finite-dimensional macrostate space introduced earlier.

In practice, this paper will instantiate deviation variables using a modest set of divergence scores and geometric distances built from benign training windows. These coordinates provide the most direct anomaly signal in the full macrostate, but their effectiveness depends on the broader context supplied by activity, structure, volatility, and persistence. A deviation score may indicate that a window is unusual; the other macro-observables help explain \emph{why} it is unusual.

\begin{remark}
This paper instantiates a practical subset of the macro-observable families above rather than every possible coordinate. The purpose of this section is to define a reusable finite macro-basis for the cyber dynamics program. The empirical study therefore selects representative and robust coordinates from each family while preserving the overall architecture of activity, disorder, structure, volatility, persistence, and deviation.
\end{remark}

Taken together, these observable families define the macrostate as a finite behavioral coordinate system rather than a bag of unrelated features. Activity measures scale and occupancy, disorder captures dispersion, structure captures organization, volatility captures local instability, persistence captures memory, and deviation captures distance from benign behavior. This decomposition is the core design choice of the paper. The next section builds on it by defining dynamics in macrostate space, including trajectories, state and transition anomaly scores, and regime-level notions of stability and recovery.
\section{Dynamics on Macrostate Space}

The preceding sections defined the macrostate as a finite-dimensional behavioral summary of a windowed cyber ensemble. That construction is necessary, but it is not sufficient for the main thesis of the paper. If the analysis stopped at feature extraction, the framework would remain only a richer anomaly-feature set. The present section therefore introduces the dynamical viewpoint that distinguishes this work from entropy-only and feature-only approaches. The central idea is that cyber behavior should be studied not only through the properties of individual windows, but through the evolution of macrostates over time: how regimes persist, how benign systems drift, how perturbations unfold, and how adversarial activity manifests as abnormal reorganization in state-space.

This perspective is motivated both by the cyber setting itself and by adjacent work on drift, temporal anomaly detection, and dynamic graph analysis. In operational environments, normal behavior is not static; it evolves across day-night cycles, maintenance periods, software updates, and user-driven workload transitions \cite{ahmed2024driftsurvey,hinder2024driftsurveyA}. Likewise, graph and temporal anomaly work has shown that abnormality often appears in evolving structure and sequence, not merely in isolated snapshots \cite{EkleEberle2024,JinEtAl2024GNN4TS}. Our framework absorbs these lessons by treating the macrostate as the basic object of dynamics and by defining anomalies both as unusual positions in macrostate space and as unusual transitions between successive macrostates.

\subsection{Macrostate Trajectories}

Let $\cT$ denote the ordered index set induced by the window family. The macrostate trajectory is
\[
\{X_t\}_{t\in \cT}.
\]
This trajectory is the primary dynamical object in the paper. Each $X_t$ summarizes the local behavior of the system in a single window, while the ordered family $\{X_t\}$ records how that behavior evolves over time. In this way, the framework moves from static summary to state-space process.

The trajectory viewpoint is important for several reasons. First, many operationally meaningful patterns are not identifiable from one window alone. A benign burst and an attack burst may yield comparable activity levels in one time slice, yet differ sharply in how they arise and how they persist. Second, trajectory analysis naturally supports regime discovery: one can ask whether the system revisits certain regions of state-space recurrently, whether it tends to remain there for long durations, and whether departures from those regions exhibit interpretable transition patterns. Third, trajectory structure allows one to distinguish \emph{drift} from \emph{disruption}. A system undergoing ordinary workload change may move smoothly through macrostate space, whereas a system under adversarial pressure may exhibit abrupt displacement, atypical oscillation, or inconsistent transition behavior.

From a practical standpoint, the trajectory may be studied at one or more temporal resolutions depending on the chosen window length and stride. Different anomaly classes may be more visible at different scales, and this is one reason the experimental design later considers multiscale windowing. Conceptually, however, the main point is independent of scale: the macrostate is not an end in itself, but a state variable whose temporal evolution is the relevant object of analysis.

\subsection{Nominal Transition Operator}

To reason about anomalous transitions, one first needs a notion of how benign macrostates are expected to evolve. We encode this through a nominal transition operator.

\begin{definition}[Nominal transition operator]
A nominal transition operator is a map
\[
T:\R^d\to \R^d
\]
or conditional expectation model
\[
T(X_t)\approx \E[X_{t+1}\mid X_t,\; X_t\in \cB].
\]
\end{definition}

The role of $T$ is to summarize nominal short-horizon evolution in macrostate space. Intuitively, if the system is currently in a benign region and operating under ordinary perturbations, then $T(X_t)$ describes where one expects the next macrostate to lie on average. The operator may be instantiated in several ways depending on model complexity and data availability. In the simplest case, $T$ may be estimated by a linear model, vector autoregressive approximation, or nearest-neighbor conditional expectation on benign windows. More flexible nonlinear or regime-conditional models are also possible. The formal definition intentionally leaves this open, because the point of the framework is not to privilege one predictor architecture, but to define a canonical transition object over macrostate space.

This definition also clarifies one of the paper's conceptual claims. A system may occupy a macrostate that is individually unsurprising, yet arrive there through a highly implausible transition. Conversely, a macrostate may appear moderately unusual in isolation but arise via smooth benign drift from a nearby nominal region. By introducing $T$, we separate the question ``where is the system?'' from the question ``how did the system move?'' In doing so, we create a formal basis for transition-aware anomaly detection.

\subsection{State and Transition Scores}

Anomaly scoring in the present framework is therefore naturally divided into state-based and transition-based components.

\begin{definition}[State anomaly score]
A state anomaly score is any functional of the form
\[
\mathcal{S}_t^{\mathrm{state}} = d(X_t,\cB),
\]
where $d$ is a suitable distance to the benign regime set.
\end{definition}

The state score measures how far the current macrostate lies from the benign operational envelope. The distance $d$ may be Euclidean after normalization, Mahalanobis-like relative to benign covariance, a manifold or cluster distance, or another geometry appropriate to the empirical benign set. The key point is that normality is represented as a set $\cB$ of regimes rather than a single reference point. This allows multiple benign modes to coexist, which is crucial in realistic cyber environments with daily cycles, workload classes, or maintenance regimes.

\begin{definition}[Transition anomaly score]
A transition anomaly score is any functional of the form
\[
\mathcal{S}_t^{\mathrm{trans}} = \norm{X_{t+1}-T(X_t)}.
\]
\end{definition}

The transition score measures departure from expected benign motion. It is thus sensitive to unusual short-horizon changes even when the current or next macrostate lies near the benign region. This is particularly useful when adversarial activity begins by perturbing the system's dynamics before displacing it far from normal operating modes. For example, the initiation of coordinated scanning, credential-use chaining, or staged beaconing may first appear as an implausible transition rather than as a fully developed outlying state. The transition score therefore operationalizes the intuition that anomaly can live in the \emph{path} and not only in the \emph{position}.

These two scores may be used separately or combined. State scores are naturally suited to detecting windows that lie outside the benign envelope, whereas transition scores are suited to detecting abnormal motion between windows. Later in the paper, the empirical evaluation will compare state-only, transition-only, and fused formulations. The theoretical value of distinguishing them is already visible here: the framework makes explicit that anomaly detection in cyber systems is both geometric and dynamical.

\subsection{Regime Transition Structure}

A further advantage of the macrostate formulation is that it supports finite regime-level descriptions of behavior. Suppose $\cB\cup \cA$ is partitioned into regimes $\cR_1,\dots,\cR_K$, either by clustering, density-based partitioning, model-based assignment, or another suitable procedure. Define the empirical regime transition matrix
\[
P_{ij} := \Prob(X_{t+1}\in \cR_j \mid X_t\in \cR_i).
\]

This matrix provides a coarse dynamical summary of the system at the regime level. It records how often the system remains in a regime, drifts into nearby regimes, or makes rare jumps into atypical ones. In benign settings, one expects recurrent operational modes to exhibit characteristic transition patterns: some regimes may be highly persistent, others may serve as transient bridges between day and night behavior, and still others may occur only during scheduled maintenance. Under attack, by contrast, one may observe low-probability transitions, uncharacteristic dwell patterns, or motion into regimes associated with coordinated reorganization.

The regime-transition viewpoint is valuable because it compresses trajectory behavior into an interpretable finite object. Instead of examining every window individually, one can study the flow of the system through a small number of macro-behavioral states. This is consistent with the structural and temporal anomaly literature, which increasingly emphasizes evolving organization over pointwise irregularity \cite{EkleEberle2024,JinEtAl2024GNN4TS}. In the present framework, however, the regime matrix is not merely descriptive; it also supports anomaly reasoning. Rare or statistically implausible regime transitions can be treated as indicators of abnormal dynamics even when raw feature changes are modest.

This regime-level description is also one of the main reasons the macrostate framework is reusable across later papers. Once a canonical state-space is defined over CSTS-compatible telemetry, regime discovery and regime-transition estimation can be reused for network, endpoint, identity, provenance, or fully heterogeneous behavioral systems without changing the surrounding logic \cite{RahmanCSTS}.

\subsection{Stability and Recovery}

The final dynamical notions introduced in this paper are stability and recovery. Even though this paper uses them only lightly, they are included now because they are central to the long-term cyber dynamics program and provide the correct language for distinguishing nominal resilience from persistent disruption.

\begin{definition}[Macrostate stability]
A regime $\cR$ is $\varepsilon$-stable over horizon $h$ if windows initialized in $\cR$
remain within an $\varepsilon$-tube of $\cR$ with high probability over $h$ future steps
under nominal perturbations.
\end{definition}

This notion formalizes the intuitive idea that benign regimes are not static points but dynamically persistent regions. A stable regime may fluctuate internally, but those fluctuations remain confined to a neighborhood of the regime under ordinary operational forces. In practice, stability can be estimated by examining dwell times, local variance, transition concentration, or empirical return properties. The main conceptual value of the definition is that it separates normal variability from genuine escape behavior. A system can be volatile and yet stable if it remains within its operational tube; conversely, a low-variance system may be unstable if small perturbations frequently eject it into abnormal regions.

\begin{definition}[Recovery time]
For an anomalous perturbation at time $t$, the recovery time is
\[
\tau_{\mathrm{rec}}(t) := \inf\{k\ge 1 : d(X_{t+k},\cB) < \varepsilon\}.
\]
\end{definition}

Recovery time measures how long it takes the system to return to the benign envelope after a perturbation. This notion is useful for at least two reasons. First, it distinguishes transient excursions from persistent anomalous occupation. A benign but abrupt operational event may temporarily displace the macrostate and then quickly relax, whereas an adversarial intrusion may induce sustained deviation or repeated re-entry into anomalous regions. Second, recovery time connects anomaly detection to resilience analysis. A system with short recovery under benign perturbations but long recovery under attacks may exhibit a meaningful separation between ordinary operational stress and malicious disruption.

Although the empirical emphasis of this paper remains anomaly discrimination rather than full resilience analysis, introducing recovery now helps establish continuity with later papers, where regime persistence, metastability, and post-perturbation dynamics will play a larger role. More broadly, the definitions of stability and recovery reinforce the central message of the paper: cyber behavior should be studied as an evolving dynamical process in a finite macrostate space, not merely as a sequence of isolated alerts or scalar entropy values.

Taken together, the constructions in this section complete the transition from feature design to dynamics. We have defined macrostate trajectories, nominal transition operators, state and transition anomaly scores, regime-transition structure, and basic notions of stability and recovery. These objects provide the formal machinery needed to treat anomaly as both a geometric and dynamical phenomenon. The next section builds on this machinery by presenting the basic theory that justifies coarse-graining, formalizes the limitations of entropy-only descriptions, and motivates multiscale, transition-aware analysis.

\section{Basic Theory}

This section records the basic theoretical statements that justify the macrostate framework introduced above. The aim is not to impose an artificially grand theory, but to formalize several simple and useful points: the macrostate map is well-defined under natural conditions; coarse-grained observables are invariant under irrelevant permutations of micro-events; entropy-only descriptions are incomplete in the presence of structural reorganization; transition-based anomaly scoring can be justified under a concentration assumption on benign dynamics; and multiscale constructions are necessary because some anomaly patterns are only visible at certain temporal resolutions. These statements are deliberately modest. Their role is to support the conceptual architecture of the paper and to clarify what the empirical study is, and is not, meant to demonstrate.

\subsection{Well-Posedness of the Macrostate Construction}

We begin with the most basic question: does the macrostate construction make mathematical sense as a map from finite event ensembles to a finite-dimensional state vector? Under the definitions already introduced, the answer is immediate, but it is useful to record the statement explicitly because later arguments rely on the existence of a well-defined coarse-graining operation.

\begin{proposition}[Well-defined macrostate map]
Suppose each coordinate of $\PhiMacro$ is a measurable finite-valued functional on finite event multisets.
Then for every finite windowed micro-ensemble $\cE_t$, the macrostate
\[
X_t=\PhiMacro(\cE_t)
\]
is well-defined.
\end{proposition}

\begin{proof}
By assumption, each coordinate of $\PhiMacro$ is defined on the class of finite event multisets and returns a finite value. Since each windowed micro-ensemble $\cE_t$ is finite by construction, every coordinate of $\PhiMacro(\cE_t)$ exists and is finite. Collecting these coordinates yields a finite-dimensional vector in $\R^d$. Hence $X_t=\PhiMacro(\cE_t)$ is well-defined.
\end{proof}

The practical meaning of this proposition is simple: once a macrostate family has been specified in terms of measurable finite-valued summaries of the event ensemble, the coarse-graining map is mathematically legitimate. This matters because the macrostate is not introduced as an informal metaphor. It is a bona fide derived object associated to each window.

The next proposition records a second basic property. Many coarse-grained observables are intended to forget irrelevant ordering details and retain only aggregate structure. If a coordinate depends only on empirical distributions, graph summaries, or temporal summary statistics, then any permutation of events that preserves those summaries should leave that coordinate unchanged.

\begin{proposition}[Permutation invariance of coarse-grained observables]
Suppose a coordinate of $\PhiMacro$ depends only on empirical distributions,
typed graph summaries, or temporal summary statistics induced by $\cE_t$.
Then that coordinate is invariant under permutations of events preserving those summaries.
\end{proposition}

\begin{proof}
Let $f$ be a coordinate of $\PhiMacro$ satisfying the stated dependence condition. If two orderings of the events in $\cE_t$ induce the same empirical distributions, the same typed graph summaries, and the same temporal summary statistics relevant to $f$, then $f$ is evaluated on identical inputs under both orderings. Therefore the resulting value is unchanged. Hence $f$ is invariant under such permutations.
\end{proof}

This proposition formalizes the intuition that macro-observables should depend on behaviorally meaningful summaries rather than on accidental record order. It is especially important for the present framework because the macrostate is intended to be a coarse-grained behavioral representation, not a fragile function of serialization artifacts or irrelevant permutations of input logs.

\subsection{Limits of Entropy-Only Descriptions}

We now formalize one of the central conceptual claims of the paper: disorder coordinates alone do not exhaust the behaviorally relevant content of a cyber window. Two micro-ensembles may have identical marginal distributions and hence identical entropy coordinates, yet differ materially in their structural organization or temporal arrangement. This is precisely why entropy is treated in this paper as one macrostate family among several rather than as the entire state description.

\begin{proposition}[Entropy-only macrostates are not structurally complete]
Let $X_t^{(H)}$ denote a macrostate retaining only distributional disorder coordinates,
and let $X_t$ denote a macrostate containing both disorder and structural-order coordinates.
Then there exist distinct micro-ensembles $\cE_t,\cE_t'$ such that
\[
X_t^{(H)} = X_t'^{(H)}
\qquad\text{but}\qquad
X_t \neq X_t'.
\]
\end{proposition}

\begin{proof}
Consider two windowed micro-ensembles $\cE_t$ and $\cE_t'$ with the same empirical category counts over the observable families used to define the disorder coordinates. For example, let both ensembles contain the same number of events per source class, destination class, service class, and protocol class, so that all empirical distributions entering $X_t^{(H)}$ are identical. Then the corresponding disorder coordinates, including Shannon or generalized entropy values computed from those distributions, are equal:
\[
X_t^{(H)} = X_t'^{(H)}.
\]

Now choose $\cE_t$ and $\cE_t'$ so that their interaction structure differs. For instance, in $\cE_t$ let communications be distributed across many source-destination pairs, while in $\cE_t'$ route the same category counts through a small number of concentrated hubs or repeated motifs. The marginal category histograms remain unchanged, but graph density, degree concentration, motif counts, or related structural summaries differ. Therefore at least one structural coordinate differs, and hence
\[
X_t \neq X_t'.
\]
This proves the claim.
\end{proof}

The key point is that entropy-like quantities collapse all configurations sharing the same chosen marginals into the same disorder summary, even when those configurations are behaviorally very different. The proposition therefore shows that entropy-only macrostates cannot be structurally complete whenever structural coordinates are admitted as behaviorally meaningful observables.

\begin{corollary}
Entropy-only detection can fail to distinguish micro-ensembles with identical
distributional disorder but different structural coordination.
\end{corollary}

\begin{proof}
By the preceding proposition, there exist distinct micro-ensembles with identical entropy-only macrostates but different full macrostates. Any detector using only $X_t^{(H)}$ assigns the same state representation to both ensembles and therefore cannot distinguish them on the basis of entropy-only information alone. If the difference between the ensembles is behaviorally meaningful, then entropy-only detection may fail.
\end{proof}

This corollary is conceptually central to the paper. It does not claim that entropy is unhelpful; rather, it proves that entropy alone cannot capture all relevant forms of cyber reorganization. That gap is what motivates the introduction of structure, persistence, volatility, and related coordinates into the macrostate.

\subsection{State and Transition Distinguishability}

The next statements justify the use of transition-aware anomaly scoring. The state score measures how far the current macrostate lies from the benign set. The transition score measures how implausible the step from $X_t$ to $X_{t+1}$ is under nominal dynamics. To make this precise, we impose a concentration assumption on benign transitions.

\begin{assumption}[Benign transition concentration]
Assume there exists a nominal transition operator $T$ and constants $\sigma,\eta>0$
such that for benign windows
\[
\Prob\!\left(\norm{X_{t+1}-T(X_t)}>\eta\right)\le \sigma.
\]
\end{assumption}

This assumption says that, under benign operation, the next macrostate is usually close to the prediction provided by the nominal transition operator. The constants $\eta$ and $\sigma$ encode, respectively, a tolerance radius and a tail probability for benign deviations. In practice, this is the natural probabilistic regularity one hopes to estimate from benign training trajectories.

\begin{proposition}[Transition anomaly criterion]
Under benign transition concentration, if
\[
\norm{X_{t+1}-T(X_t)}>\eta,
\]
then the transition is anomalous at confidence level controlled by $\sigma$.
\end{proposition}

\begin{proof}
By the benign transition concentration assumption,
\[
\Prob\!\left(\norm{X_{t+1}-T(X_t)}>\eta \,\middle|\, X_t\in\cB\right)\le \sigma.
\]
Therefore any observed transition satisfying
\[
\norm{X_{t+1}-T(X_t)}>\eta
\]
falls into an event whose benign probability is at most $\sigma$. Hence the transition is anomalous at confidence level controlled by $\sigma$.
\end{proof}

This proposition is deliberately modest: it does not prove optimality of any detector, only that large transition residuals are statistically exceptional under the assumed benign model. That is exactly the level of justification needed for the present paper.

We now state the assumption that captures the possibility of adversarial reorganization that escapes entropy-only characterization.

\begin{assumption}[Structured adversarial reorganization]
Assume an adversarial perturbation may preserve selected marginal distributions
while inducing a nontrivial change in structural or coupling coordinates.
\end{assumption}

This assumption is well aligned with many cyber scenarios. An attacker may preserve overall traffic volumes or even category histograms while rerouting communication through a small number of hubs, establishing persistent low-volume beaconing channels, concentrating role transitions, or otherwise reorganizing the relational structure of the system.

\begin{theorem}[Distinguishability beyond entropy]
Under structured adversarial reorganization, there exist attack windows for which
entropy-only state scores fail to separate benign from adversarial behavior, while
macrostate scores using structural or coupling coordinates succeed.
\end{theorem}

\begin{proof}
Under the structured adversarial reorganization assumption, choose an attack window whose perturbation preserves the marginal distributions used in the entropy-only macrostate $X_t^{(H)}$ but changes at least one structural or coupling coordinate. By construction, the entropy-only representation of the adversarial window coincides with that of a benign window sharing the same selected marginals. Hence any state score depending only on $X_t^{(H)}$ assigns the same entropy-only value to both windows and therefore fails to separate them.

On the other hand, because at least one structural or coupling coordinate changes, the full macrostate differs between the benign and adversarial windows. Consequently, any macrostate score that depends on a distance in the full state-space and assigns positive weight to the changed structural or coupling subspace can separate the two windows. Therefore there exist attack windows for which entropy-only scores fail while fuller macrostate scores succeed.
\end{proof}

The theorem formalizes the main interpretive payoff of the macrostate framework. If adversarial behavior can manifest as structured reorganization without strong marginal redistribution, then entropy-only state descriptions are provably insufficient in those cases, while a broader behavioral state-space has the capacity to distinguish them.

\subsection{Multiscale Necessity}

The final result addresses temporal scale. The preceding framework is defined for a fixed window length $\Delta$, but in practice no single scale is universally sufficient. Some anomalies are visible only when the analyst looks closely; others emerge only after longer accumulation. This motivates the need for multiscale macrostate families.

\begin{proposition}[Single-scale insufficiency]
There exist anomaly patterns detectable at one window scale $\Delta_1$ and not at
another scale $\Delta_2$.
\end{proposition}

\begin{proof}
Consider first a short, coordinated burst anomaly whose duration is much smaller than $\Delta_2$ but comparable to $\Delta_1$, with $\Delta_1<\Delta_2$. At scale $\Delta_1$, the anomalous burst occupies a substantial portion of the window and significantly alters the macro-observables, making detection possible. At scale $\Delta_2$, the same burst is averaged together with a much larger amount of benign activity, potentially diluting its effect below detectability.

Conversely, consider a slow drift anomaly whose effect accumulates gradually over a horizon much longer than $\Delta_1$ but is visible over $\Delta_2$. At the shorter scale $\Delta_1$, each local window may remain close to benign behavior, so the anomaly is not distinguishable. At the longer scale $\Delta_2$, the cumulative shift changes the macrostate enough to be detected.

Thus there exist anomaly patterns that are detectable at one scale and not at another.
\end{proof}

This proposition is basic but operationally important. It shows that scale is not merely a tuning parameter; it is part of the representational problem itself. A fixed window size can hide some behaviors and reveal others.

\begin{corollary}[Value of multiscale macrostate families]
A multiscale macrostate family strictly dominates any fixed single-scale macrostate
on classes of anomalies exhibiting scale-dependent visibility.
\end{corollary}

\begin{proof}
By the preceding proposition, there exist anomaly classes detectable at $\Delta_1$ but not $\Delta_2$, and vice versa. A multiscale family including both scales can detect the union of these classes, whereas either single-scale family misses at least one class. Hence the multiscale family strictly dominates any fixed single-scale macrostate on anomaly classes with scale-dependent visibility.
\end{proof}

Taken together, the results of this section provide the formal backbone of the paper. The macrostate construction is well-defined; coarse-grained observables are invariant under irrelevant permutations; entropy-only descriptions are incomplete in the presence of structural reorganization; transition anomaly scores are justified under benign concentration assumptions; and multiscale constructions are necessary because anomaly visibility depends on temporal resolution. None of these claims is overly strong, but together they justify the paper's central design choices. The next section turns from theory to experimental design, where these ideas are instantiated and tested on benchmark network telemetry.

\section{Experimental Design}

This section specifies the experimental plan for evaluating the proposed macrostate framework on benchmark network telemetry. The guiding principle is that the experiments should test the central claims of the paper directly: whether a finite macrostate representation improves anomaly discrimination beyond entropy-only baselines, whether transition-aware scoring adds value beyond state-only scoring, whether structural and temporal macro-variables materially improve separability, and whether multiscale analysis is necessary for robust detection. Because the paper is explicitly about dynamics, the experimental design preserves temporal order throughout preprocessing, windowing, training, validation, and testing. In particular, we do not randomize away the temporal structure of the underlying telemetry, since doing so would undermine the state-transition viewpoint introduced earlier.

\subsection{Research Questions}

The empirical study is organized around the following research questions.

\begin{itemize}[leftmargin=2em]
    \item \textbf{RQ1.} Does a finite macrostate representation outperform entropy-only baselines for anomaly detection in network telemetry?
    
    \item \textbf{RQ2.} Do transition-based anomaly scores improve discrimination between benign workload drift and adversarial change?
    
    \item \textbf{RQ3.} Which macro-variable families contribute most to separability?
    
    \item \textbf{RQ4.} How sensitive is performance to window scale?
    
    \item \textbf{RQ5.} Can macrostate trajectories reveal interpretable behavioral regimes?
\end{itemize}

These questions are ordered intentionally. RQ1 establishes whether the macrostate idea has practical value at all. RQ2 tests the specifically dynamical component of the framework. RQ3 clarifies whether the gains come mainly from structural, temporal, or persistence-related coordinates rather than from simple dimensional expansion. RQ4 tests the theory claim that single-scale analysis is insufficient. RQ5 addresses the interpretability goal of the paper by asking whether the induced state-space supports useful regime-level analysis beyond scalar detection.

\subsection{Datasets}

This paper focuses on benchmark network telemetry datasets that provide sufficient event structure to construct time-ordered windows and evaluate anomaly discrimination under varying attack types. The primary datasets are as follows.

\begin{itemize}[leftmargin=2em]
    \item \textbf{UNSW-NB15.} Used as a general-purpose network intrusion benchmark with mixed benign and attack traffic and a broad range of traffic attributes.
    
    \item \textbf{CIC-IDS2017} (or a closely related CIC intrusion dataset). Used as the main benchmark for time-ordered network attack behavior across multiple attack categories.
    
    \item \textbf{CICIoT2023} (optional, if experiment capacity permits). Used as an additional robustness dataset emphasizing high-volume IoT-oriented traffic and attack diversity.
\end{itemize}

For each dataset, we record the following properties in the final implementation package:
\begin{itemize}[leftmargin=2em]
    \item telemetry type and native record schema,
    \item timestamp availability and temporal continuity,
    \item attack label granularity,
    \item attack family categories,
    \item number of records and recording span,
    \item train/validation/test partition plan.
\end{itemize}

The split strategy preserves temporal order. We do \emph{not} randomly shuffle records prior to window construction. Instead, for each dataset we create chronologically ordered event streams and then partition windows into training, validation, and test sets using contiguous temporal segments. Benign-only training is preferred for state-baseline and transition-baseline estimation, while validation and test segments include both benign and anomalous windows. When attack labels are sparse or temporally localized, segment boundaries are chosen so that validation and test sets contain sufficient attack coverage without contaminating benign-only training baselines.

If a dataset is too fragmented temporally to support clean transition modeling, we use it only for state-only comparisons and document that limitation explicitly. The primary emphasis remains on datasets that preserve enough sequential structure for macrostate dynamics to be meaningful.

\subsection{Preprocessing and Event Normalization}

The preprocessing pipeline converts each dataset into a unified event stream compatible with the microstate formalism introduced earlier. The exact steps are:

\begin{enumerate}[leftmargin=2em,label=\arabic*)]
    \item \textbf{Schema normalization.} Convert each raw record into a common network-event schema with fields corresponding to timestamp, source identifier, destination identifier, service or port class, protocol, direction, and event attributes such as bytes, packets, and duration.
    
    \item \textbf{Role mapping.} Map source and destination identifiers into coarse role categories whenever possible, such as internal, external, server, client, subnet class, or dataset-specific host role. If full role labeling is not available, use robust derived proxies such as internal/external partitioning, source/destination frequency class, or service-host behavior type.
    
    \item \textbf{Service normalization.} Bucket ports or service identifiers into service classes. At minimum this includes well-known ports, common protocol families, and an ``other'' bucket. If application-layer labels are available, they are folded into the same service-class framework.
    
    \item \textbf{Numerical binning.} Bin bytes, packets, and durations into stable coarse ranges for empirical distribution construction. Both quantile-based and log-scale binning are acceptable, but the binning strategy must be fixed per dataset prior to model training.
    
    \item \textbf{Chronological ordering.} Sort events by timestamp and construct a single ordered event stream per recording segment.
    
    \item \textbf{Window family construction.} Build time windows at multiple scales from the ordered stream and assign each event to the windows it overlaps.
\end{enumerate}

All preprocessing outputs are serialized in reusable intermediate files so that feature extraction, baseline fitting, and ablation runs can be reproduced without re-running the full ingestion pipeline.

For class imbalance, we do \emph{not} rebalance at the raw event level. Instead, we preserve the true window-level incidence of anomalies and handle imbalance at the scoring and evaluation stage through AUPRC, threshold calibration, and, where appropriate, class-weighted shallow supervised models. For unsupervised and one-class methods, benign-only training remains the default.

Label alignment is handled at the window level. Event-level attack labels, session labels, or flow labels are converted into window labels by overlap rules specified below. This is necessary because the macrostate is defined on windows rather than on individual events.

\subsection{Windowing Plan}

The core experimental window family uses the following window lengths:
\[
\Delta \in \{30\text{s}, 1\text{m}, 5\text{m}, 15\text{m}\}.
\]
The default stride is
\[
\delta = \Delta/2.
\]
This provides overlapping windows while keeping computational cost moderate.
The choice is intentional: because the paper studies macrostate trajectories
rather than isolated window scores, the stride must be small enough to expose
boundary-crossing behavioral shifts.  A disjoint choice $\delta=\Delta$ can
miss short-lived onset behavior or split a transition across adjacent windows
in a way that weakens the induced anomaly signal.  The half-window stride
therefore gives a practical compromise between trajectory fidelity and
computational cost.

As a robustness study, one selected dataset subset will also be evaluated with
\[
\delta = \Delta/4
\]
to assess whether denser trajectory sampling materially changes the state and transition results.

Window labels are assigned by overlap with malicious events. Let $\rho$ denote the minimum fraction of window activity associated with maliciously labeled records required to classify the window as anomalous. We perform a sensitivity study over
\[
\rho \in \{0.05, 0.25, 0.50\}.
\]
Operationally, this means:
\begin{itemize}[leftmargin=2em]
    \item a low threshold $\rho=0.05$ captures early contamination and is useful for detection-delay analysis,
    \item an intermediate threshold $\rho=0.25$ provides a balanced labeling criterion,
    \item a high threshold $\rho=0.50$ enforces stricter anomaly occupancy and tests robustness to label ambiguity.
\end{itemize}

For datasets whose labels are session-level rather than event-level, the overlap criterion is computed using record counts or duration overlap, depending on what the dataset supports. The chosen rule must be documented explicitly per dataset.

This windowing plan is directly aligned with the theory section. Short windows help capture brief coordinated bursts; medium windows support local regime characterization; and longer windows capture slower drift and sustained reorganization. The multiscale family is therefore not merely a robustness convenience, but a test of the paper's single-scale insufficiency claim.

\subsection{Macro-Variable Instantiation}

Although the full macrostate formalism allows many coordinates, this paper adopts a moderate, reproducible instantiation designed to test the main hypothesis without excessive feature proliferation. The first-paper macrostate is:

\paragraph{Activity family \(A_t\).}
\begin{itemize}[leftmargin=2em]
    \item total flow count,
    \item total byte count,
    \item unique source count,
    \item unique destination count,
    \item unique service count.
\end{itemize}

\paragraph{Disorder family \(H_t\).}
\begin{itemize}[leftmargin=2em]
    \item Shannon entropy over source-role distribution,
    \item Shannon entropy over destination-role distribution,
    \item Shannon entropy over service distribution,
    \item Shannon entropy over protocol distribution,
    \item Shannon entropy over byte-bin distribution,
    \item optional one R\'enyi entropy coordinate,
    \item optional one Tsallis entropy coordinate.
\end{itemize}

\paragraph{Structural family \(S_t\).}
\begin{itemize}[leftmargin=2em]
    \item interaction graph density,
    \item source-degree concentration,
    \item destination-degree concentration,
    \item service-access concentration,
    \item simple motif counts such as fan-out, many-to-one, and repeated-pair motifs.
\end{itemize}

\paragraph{Volatility family \(V_t\).}
\begin{itemize}[leftmargin=2em]
    \item sub-window burstiness of flow count,
    \item within-window variance of sub-bin flow counts,
    \item within-window variance of sub-bin byte totals.
\end{itemize}

\paragraph{Persistence family \(P_t\).}
\begin{itemize}[leftmargin=2em]
    \item repeated edge ratio relative to prior window,
    \item recurring source-destination pair ratio,
    \item recurring source-service path ratio.
\end{itemize}

\paragraph{Deviation family \(D_t\).}
\begin{itemize}[leftmargin=2em]
    \item Jensen--Shannon divergence to benign baseline histogram,
    \item optional geometric distance to benign macrostate centroid or benign cluster manifold.
\end{itemize}

This work deliberately omits a large explicit coupling family \(C_t\) in the first empirical pass unless a dataset supports clean subsystem partitioning. If coupling variables can be instantiated robustly, they may be added as an extension, but they are not required for the core first-paper experiments.

All macro-variables are standardized using statistics computed on benign training windows only. This prevents attack contamination of scale parameters and keeps the evaluation aligned with the benign-baseline philosophy of the framework.

\subsection{Baselines}

We compare the proposed macrostate representation against three baseline families.

\paragraph{Baseline Family A: Entropy-only detectors.}
\begin{itemize}[leftmargin=2em]
    \item threshold detector using one or more Shannon entropy coordinates,
    \item shallow detector using the vector of entropy features alone,
    \item R\'enyi/Tsallis variants when computationally feasible.
\end{itemize}

The threshold detector serves as the simplest classical comparison. The shallow entropy-vector detector tests whether the gains of the proposed framework are merely due to using multiple disorder coordinates rather than a single entropy scalar.

\paragraph{Baseline Family B: Conventional anomaly detectors.}
\begin{itemize}[leftmargin=2em]
    \item Isolation Forest on standard aggregate features,
    \item One-Class SVM on standard aggregate features,
    \item autoencoder on standard aggregate features (optional, if time permits).
\end{itemize}

The standard aggregate feature set consists of simple count and volume summaries similar to what many baseline systems would use operationally, but without the full macrostate organization.

\paragraph{Baseline Family C: Proposed macrostate detectors.}
\begin{itemize}[leftmargin=2em]
    \item detector on \(H_t\) only,
    \item detector on \((A_t,H_t)\),
    \item detector on \((A_t,H_t,S_t,V_t)\),
    \item detector on full state \((A_t,H_t,S_t,V_t,P_t,D_t)\).
\end{itemize}

These staged models are the core ablation ladder and make it possible to test where the gains actually come from.

For transition modeling, the candidate nominal transition operators are:
\begin{itemize}[leftmargin=2em]
    \item naive persistence: \(T(X_t)=X_t\),
    \item linear regression from \(X_t\) to \(X_{t+1}\),
    \item VAR-style transition model on selected coordinates,
    \item optional small recurrent model if the simpler models are clearly inadequate.
\end{itemize}

The default plan is to begin with naive persistence and linear regression, since these are easier to interpret and sufficient to test whether transition-aware scoring has empirical value.

\subsection{Evaluation Metrics}

The primary evaluation metrics are:
\begin{itemize}[leftmargin=2em]
    \item AUROC,
    \item AUPRC,
    \item F1 and best-F1 over threshold sweep,
    \item false positive rate at fixed recall,
    \item detection delay / time-to-detect.
\end{itemize}

AUPRC is especially important because anomalous windows are typically sparse relative to benign windows. Detection delay is included because the state-transition framework is intended not only to classify windows, but to identify emerging anomalies early.

If regime discovery is used, we also report:
\begin{itemize}[leftmargin=2em]
    \item regime purity,
    \item mean dwell time,
    \item transition concentration,
    \item benign versus anomalous regime occupancy proportions.
\end{itemize}

Interpretability outputs include:
\begin{itemize}[leftmargin=2em]
    \item macrostate trajectory plots,
    \item regime maps in reduced dimension,
    \item regime-transition heatmaps,
    \item persistence and recovery curves where applicable.
\end{itemize}

The goal is to ensure that the paper is not judged solely on scalar metrics. Since one of the main claims is interpretability at the regime level, the visual and structural outputs are first-class results, not just auxiliary illustrations.

\subsection{Ablation Study Plan}

The feature-family ablation ladder is:

\begin{itemize}[leftmargin=2em]
    \item \textbf{Ablation 1:} \(H\) only,
    \item \textbf{Ablation 2:} \(A+H\),
    \item \textbf{Ablation 3:} \(A+H+V\),
    \item \textbf{Ablation 4:} \(A+H+S+V\),
    \item \textbf{Ablation 5:} full state \(A+H+S+V+P+D\).
\end{itemize}

This ladder is designed to answer RQ3 directly. It shows whether structural and persistence variables provide real gains beyond entropy and load alone.

The transition ablation includes:
\begin{itemize}[leftmargin=2em]
    \item state-only score,
    \item transition-only score,
    \item fused state-transition score.
\end{itemize}

The fused score can be implemented as a validation-tuned weighted combination or as a shallow meta-detector over the two scores.

The scale ablation includes:
\begin{itemize}[leftmargin=2em]
    \item each window length alone,
    \item multiscale fusion across all selected window lengths.
\end{itemize}

Multiscale fusion may be implemented by concatenating per-scale scores, taking a max-score rule, or training a shallow fusion layer on per-scale macrostate summaries. The simplest robust option is to begin with score-level fusion before attempting full feature-level multiscale concatenation.

\subsection{Statistical Testing}

All reported metrics are computed across multiple runs or temporal folds whenever the dataset structure permits. We report means together with confidence intervals, preferably via nonparametric bootstrap at the window level or fold level depending on the evaluation setting.

For key pairwise comparisons, especially entropy-only versus full macrostate detectors, we perform paired significance testing on per-split performance values. Suitable tests include:
\begin{itemize}[leftmargin=2em]
    \item paired bootstrap confidence intervals for metric differences,
    \item paired permutation tests,
    \item Wilcoxon signed-rank tests on split-level scores when appropriate.
\end{itemize}

The primary inferential target is not broad null-hypothesis theater, but a disciplined answer to the practical question: are the gains of the full macrostate framework over entropy-only and standard aggregate baselines consistent across datasets, scales, and split configurations?

Taken together, this experimental design gives an exact and executable plan for evaluating the macrostate framework. It is sufficiently constrained to support reproducibility, sufficiently modular to support ablation and robustness studies, and sufficiently aligned with the theoretical claims of the paper to make the eventual empirical results interpretable. 

\section{Results}

This section reports the first empirical results of the proposed finite macrostate framework on benchmark network telemetry. The clearest current evidence comes from the cleaned UNSW-NB15 experiments, with the strongest and most interpretable results appearing at the $60\,$s scale. Bounded first-pass experiments on CIC-IDS2017 at $30\,$s and $60\,$s reproduce the same directional pattern, namely that richer macrostate families improve over entropy-only baselines while the most useful macrostate composition can vary with temporal scale. Accordingly, UNSW-NB15 serves as the primary empirical anchor of this paper, while CIC-IDS2017 provides bounded replication rather than a fully symmetric second benchmark campaign.

We emphasize three points at the outset. First, the present paper is about validating the finite macrostate representation itself, not yet about claiming that every layer of the broader cyber-dynamics program is fully mature. Second, the UNSW-NB15 results should be read as the primary empirical anchor because they support cleaned split construction, scale analysis, and multiscale comparisons. Third, the CIC-IDS2017 results should be read as bounded replication evidence: they are sufficient to test whether the macrostate advantage survives on a second benchmark, but they are not yet intended as a fully symmetric all-scale second campaign.

\begin{table}[t]
\centering
\caption{Primary results at informative scales. Entries report AUPRC.}
\label{tab:primary_results_compact}
\renewcommand{\arraystretch}{1.08}
\setlength{\tabcolsep}{3.5pt}
\footnotesize
\begin{tabular}{llll}
\hline
Dataset & Scale & Best baseline & Best macrostate \\
\hline
UNSW & 30s & OCSVM, 0.8492 & Full, 0.8569 \\
UNSW & 60s & OCSVM, 0.9631 & Full, 0.9730 \\
CIC & 30s & EV-LR, 0.7387 & A+H+V, 0.7797 \\
CIC & 60s & EV-LR, 0.7672 & Full, 0.7917 \\
\hline
\end{tabular}
\end{table}

In Table~\ref{tab:primary_results_compact}, OCSVM denotes One-Class SVM, EV-LR denotes entropy-vector logistic regression, and Full denotes the full macrostate $A{+}H{+}S{+}V{+}P{+}D$.

\subsection{UNSW-NB15: Cleaned Primary Results}

The cleaned UNSW-NB15 experiments provide the strongest evidence in the paper. After repairing the evaluation split to avoid label-degenerate large-window test regions, all four scales became formally viable, but the most informative scales remained $30\,$s and $60\,$s. The larger $300\,$s and $900\,$s settings are still more imbalanced and are therefore treated as secondary evidence rather than as the primary basis for interpretation.

At the baseline level, entropy-only detection remains meaningful but is not dominant. At both $30\,$s and $60\,$s, strong baseline performance is achievable, especially with One-Class SVM, but once richer macrostate families are introduced the entropy-only representation is no longer best. The most important UNSW result appears at the $60\,$s scale, where the full macrostate achieves the strongest clean performance, with AUPRC approximately $0.9730$, exceeding the best entropy-only baseline. At $30\,$s, the same directional pattern holds: the full macrostate again outperforms entropy-only formulations, though with a smaller margin than at $60\,$s.

These results support the central empirical claim of this paper. Entropy remains useful, but entropy-only descriptions are not sufficient as complete behavioral summaries of cyber windows. A richer macrostate that combines activity, disorder, structure, volatility, persistence, and deviation provides a stronger representation of anomaly structure than scalar disorder alone.

\subsection{Macrostate Ablation Results on UNSW}
An ablation study systematically removes or adds parts of a model or representation to measure how much each component contributes to performance. In this paper, we use ablation to show that the macrostate gains are not accidental: entropy helps, but activity, structure, volatility, persistence, and deviation each add interpretable value, which helps justify why a richer behavioral state-space outperforms entropy-only formulations.

The ablation ladder in this study uses five feature families: $H$, $A{+}H$, $A{+}H{+}V$, $A{+}H{+}S{+}V$, and the full macrostate $A{+}H{+}S{+}V{+}P{+}D$.

\begin{table}[t]
\centering
\caption{Best ablation stage by dataset and scale.}
\label{tab:ablation_compact}
\renewcommand{\arraystretch}{1.08}
\setlength{\tabcolsep}{4pt}
\footnotesize
\begin{tabular}{llll}
\hline
Dataset & Scale & Best stage & AUPRC \\
\hline
UNSW & 30s & Full & 0.8569 \\
UNSW & 60s & Full & 0.9730 \\
CIC & 30s & A+H+V & 0.7797 \\
CIC & 60s & Full & 0.7917 \\
\hline
\end{tabular}
\end{table}

Several consistent patterns emerge on UNSW. First, moving from $H$ alone to $A{+}H$ provides a clear gain, confirming that disorder without workload context is incomplete. Second, the addition of structural variables yields a further improvement, and this improvement is consistent across the informative scales. Third, persistence and deviation variables add a smaller but positive final lift over the already-strong $A{+}H{+}S{+}V$ configuration. In other words, the empirical ladder mirrors the theoretical argument of the paper: entropy contributes useful information, but the strongest representation arises when disorder is embedded in a broader behavioral state-space.

This ablation pattern is one of the most important empirical outputs of this paper. It shows that the macrostate is not simply a larger feature bag whose gains are opaque. Instead, the gains align with an interpretable progression: entropy helps, activity contextualizes it, structure strengthens it, and persistence/deviation sharpen the final discrimination.

\begin{figure}[t]
    \centering
    \includegraphics[width=\linewidth]{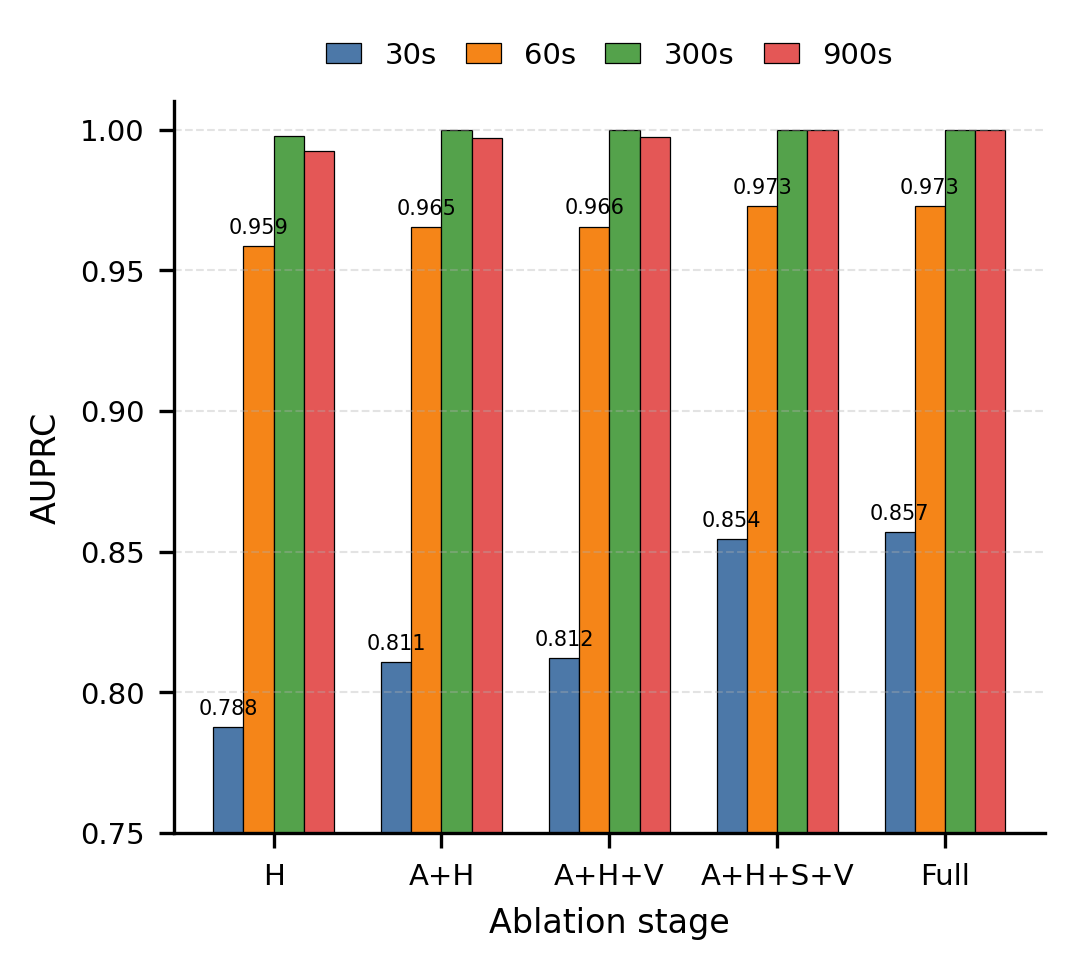}
    \caption{UNSW-NB15 ablation results across window scales, reported in AUPRC. The x-axis corresponds to the ablation ladder $H$, $A{+}H$, $A{+}H{+}V$, $A{+}H{+}S{+}V$, and Full; the legend denotes the four window scales. The strongest clean evidence appears at $60\,$s, where the full macrostate is best.}
    \label{fig:unsw_ablation_by_scale}
\end{figure}

Figure~\ref{fig:unsw_ablation_by_scale} makes the ablation pattern visually explicit. At the informative scales, performance improves as the representation moves from entropy-only toward richer behavioral macrostates. The largest and clearest gain appears when structural variables are introduced, while persistence and deviation provide a smaller but still positive final lift. The strongest clean pattern appears at $60\,$s, where the full macrostate is best. This is important because it shows that the macrostate gains are not arbitrary; they arise through an interpretable progression from disorder alone toward a broader behavioral state description. At the informative scales, performance improves as the representation moves from entropy-only toward richer behavioral macrostates. The largest and clearest gain appears when structural variables are introduced, while persistence and deviation provide a smaller but still positive final lift. This is important because it shows that the macrostate gains are not arbitrary; they arise through an interpretable progression from disorder alone toward a broader behavioral state description.

\subsection{State Versus Transition Scoring on UNSW}

The transition-aware experiments yield a more cautious result than the state-based ones.

\begin{table}[t]
\centering
\caption{Transition-scoring summary.}
\label{tab:transition_compact}
\renewcommand{\arraystretch}{1.08}
\setlength{\tabcolsep}{4pt}
\footnotesize
\begin{tabular}{llll}
\hline
Dataset & Scale & Best mode & Note \\
\hline
UNSW & 30s & State-only & Transition weaker \\
UNSW & 60s & State-only & Fusion no gain \\
CIC & 30s & --- & Deferred \\
CIC & 60s & --- & Deferred \\
\hline
\end{tabular}
\end{table}

Across the informative UNSW scales, state-only scoring remains stronger than transition-only scoring, and simple fused state-transition scoring does not improve over the state-only formulation. This means that this paper should not claim a demonstrated transition benefit. Instead, the correct interpretation is narrower: the finite macrostate representation itself is already useful, while the particular transition models used in this first paper are too simple to unlock the full dynamical benefit of the framework.

This is an informative negative result rather than a failure of the overall program. It suggests that the main value of this paper lies in establishing a reusable behavioral state-space, while stronger exploitation of trajectories, transition concentration, regime escape, and recovery belongs more naturally to later papers.

\subsection{Scale Sensitivity on UNSW}

The multiscale experiments support the claim that anomaly visibility depends materially on temporal resolution. The $60\,$s scale provides the strongest clean evidence overall on UNSW, while $30\,$s is also informative and directionally consistent. By contrast, the larger $300\,$s and $900\,$s scales remain more imbalanced even after split repair, which makes their headline ranking metrics less trustworthy as primary evidence. The correct interpretation is therefore not that larger windows are universally better, but that temporal scale changes what is visible and how reliably it can be estimated.

\begin{figure}[t]
    \centering
    \includegraphics[width=\linewidth]{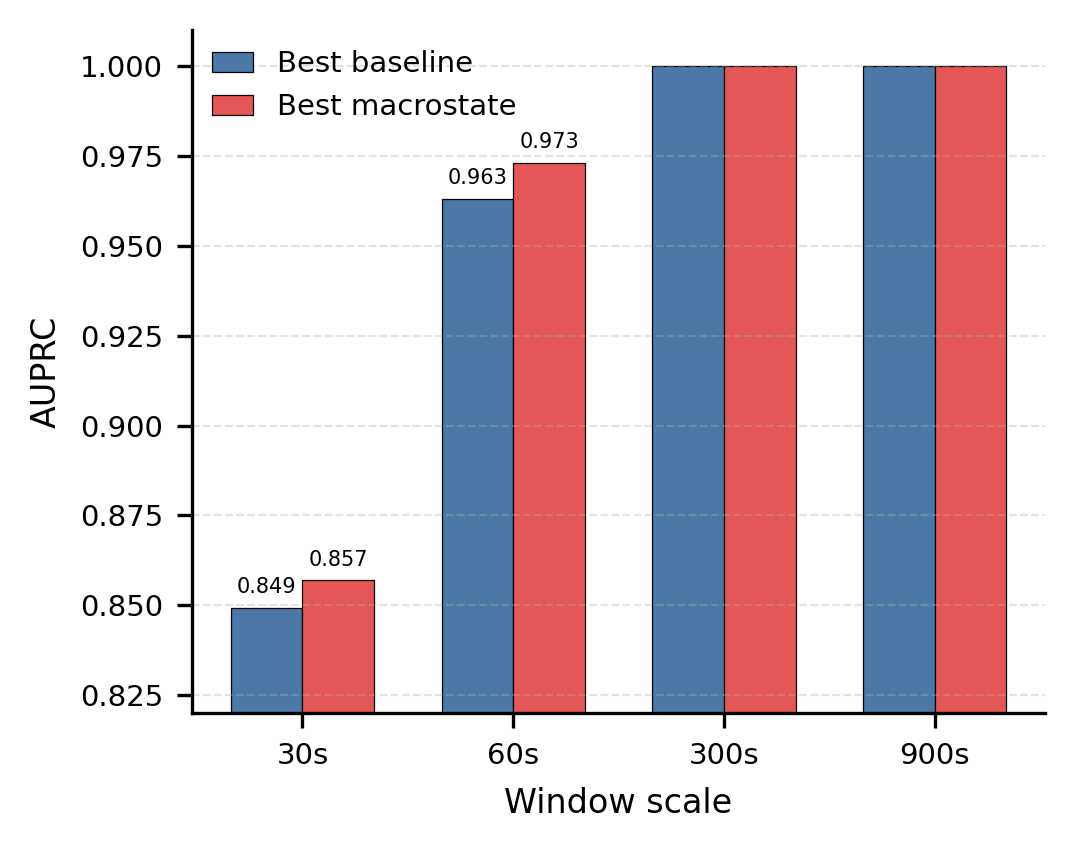}
    \caption{Scale sensitivity on UNSW-NB15 under the cleaned split policy. The figure compares the best baseline and the best macrostate model at each window scale using AUPRC. The 60s scale provides the strongest clean evidence, while larger scales behave differently and remain more class-imbalanced even after split repair.}
    \label{fig:unsw_scale_sensitivity}
\end{figure}

Figure~\ref{fig:unsw_scale_sensitivity} shows that scale is not a minor tuning parameter but part of the representational problem itself. Accordingly, the paper's primary empirical interpretation is anchored on the $30\,$s and $60\,$s scales rather than on the larger-window settings. The $60\,$s scale yields the strongest clean result, while $30\,$s remains informative and directionally consistent. The larger $300\,$s and $900\,$s scales are shown for completeness and as secondary evidence, but their greater imbalance means that they should not carry the primary interpretive weight of the paper. This figure therefore supports the claim that anomaly visibility depends materially on temporal resolution. The $60\,$s scale yields the strongest clean result, while $30\,$s remains informative and directionally consistent. The larger scales are still useful as secondary evidence, but their greater imbalance means that they should not carry the primary interpretive weight of the paper. This figure therefore supports the claim that anomaly visibility depends materially on temporal resolution.

A first multiscale fusion experiment was also performed. Simple score-level fusion improves over the baseline regime and is therefore directionally useful, but it does not exceed the best single-scale full macrostate result at $60\,$s. This means that multiscale analysis is justified as a representational necessity, but multiscale fusion is not yet empirically superior in its current simple form.

\begin{figure}[t]
    \centering
    \includegraphics[width=\linewidth]{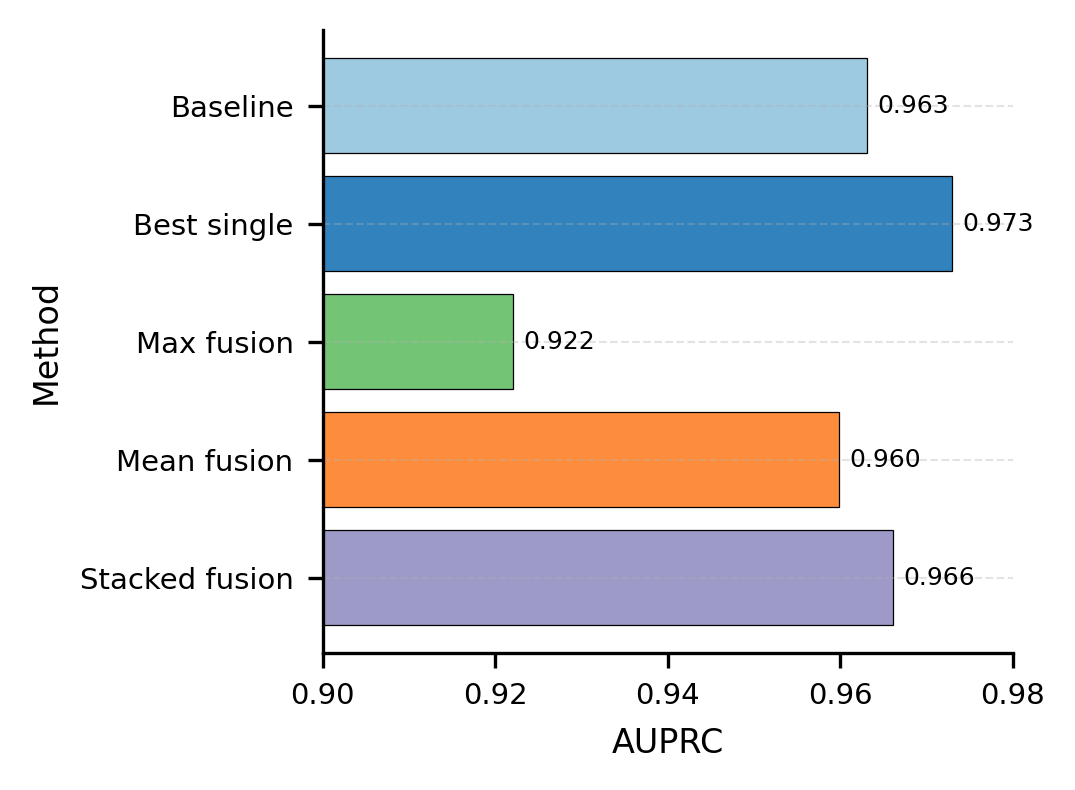}
    \caption{Multiscale fusion results on UNSW-NB15. Simple score-level and stacked fusion improve over the baseline regime, but they do not surpass the best single-scale macrostate model. This indicates that multiscale aggregation is useful, but not yet stronger than the best informative single-scale representation in the current setup.}
    \label{fig:unsw_multiscale_fusion}
\end{figure}

Figure~\ref{fig:unsw_multiscale_fusion} clarifies the role of multiscale fusion in the current paper. The multiscale combinations improve over weaker comparison points and over the broader baseline regime, which is encouraging, but they do not surpass the best single-scale full macrostate at $60\,$s. The correct conclusion is therefore that multiscale analysis is representationally justified, while multiscale fusion remains preliminary and is not yet empirically superior to the best informative single-scale representation. The multiscale combinations improve over the baseline regime, which is encouraging, but they do not surpass the best single-scale full macrostate at $60\,$s. The correct conclusion is therefore that multiscale analysis is representationally justified, while effective multiscale fusion remains an open empirical question.

\subsection{Macrostate Geometry on UNSW}

Reduced-dimensional geometry plots further support the behavioral-state interpretation of the framework. PCA visualizations of the macrostate space show visible benign-versus-malicious separation, especially at the informative scales. At $60\,$s, the first two principal components already capture nontrivial separation between benign and anomalous windows, consistent with the view that the macrostate is not merely useful for scalar classification but also organizes behavior into interpretable geometric structure.

\begin{figure}[t]
    \centering
    \includegraphics[width=\linewidth]{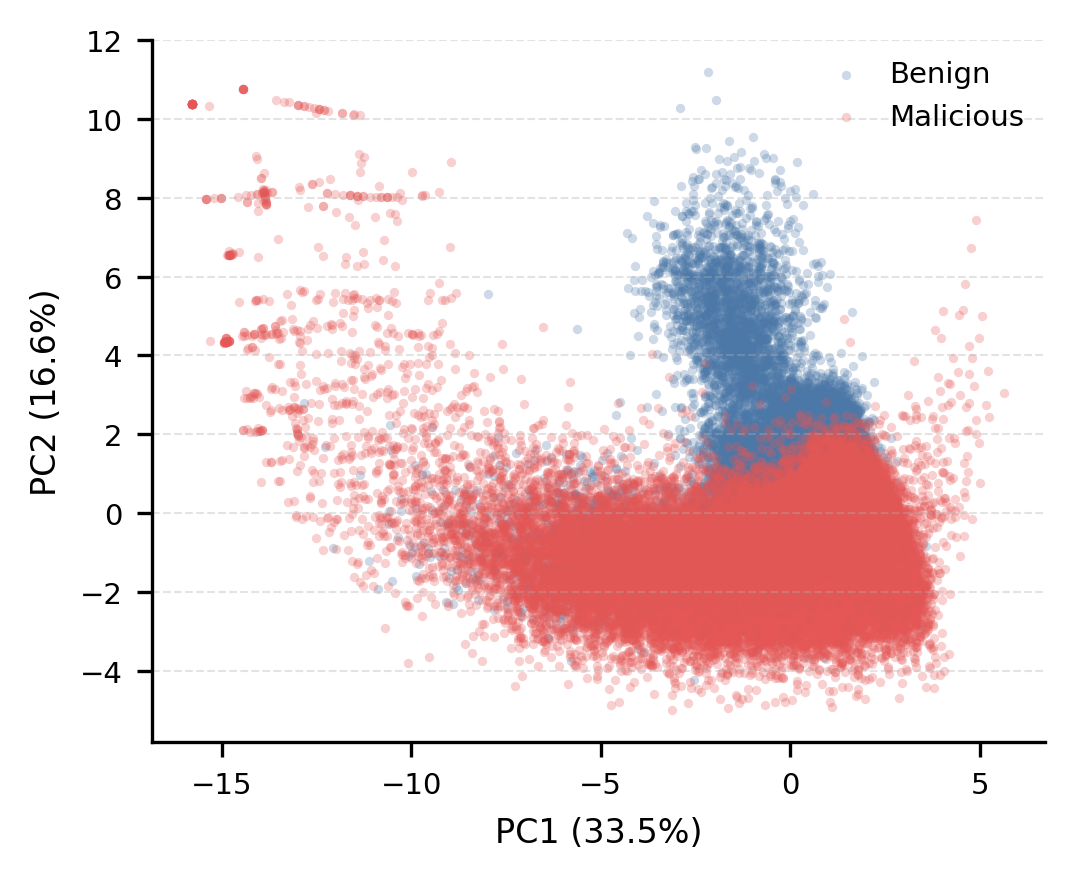}
    \caption{Two-dimensional PCA projection of the 60s UNSW macrostate representation. Benign and malicious windows exhibit partial but visible separation, supporting the claim that the learned macrostate induces a meaningful state-space rather than collapsing the classes into a single overlapping distribution.}
    \label{fig:unsw_pca_60s}
\end{figure}

Figure~\ref{fig:unsw_pca_60s} provides a geometric view of the macrostate space at the strongest clean scale. The benign and malicious windows are not perfectly separable in two dimensions, but they do occupy visibly different regions of the projected state-space. The projection is used interpretively rather than as evidence of linear separability. This matters because the macrostate is intended not only to improve scalar detection performance, but also to support interpretable behavioral-state analysis of cyber behavior. The benign and malicious windows are not perfectly separable in two dimensions, but they do occupy visibly different regions of the projected state-space. This matters because the macrostate is intended not only to improve scalar detection performance, but also to support interpretable behavioral-state reasoning about cyber behavior.

\subsection{CIC-IDS2017: Bounded Replication at 30s and 60s}

The CIC-IDS2017 experiments were designed as bounded replication runs rather than as a full second all-scale campaign. To keep runtime under control, we implemented fast deviation patches and ran first-pass experiments at $60\,$s and then $30\,$s only. Even under this bounded setup, the directional macrostate result replicated.

At $60\,$s, the best entropy-only baseline is entropy-vector logistic regression, with AUPRC approximately $0.7672$. The best macrostate stage is the full macrostate $A{+}H{+}S{+}V{+}P{+}D$, with AUPRC approximately $0.7917$, providing a modest but real gain. At $30\,$s, the same directional effect remains: entropy-vector logistic regression is again the strongest baseline, while the best macrostate stage outperforms it. Importantly, however, the strongest macrostate at $30\,$s is not the full state but the intermediate $A{+}H{+}V$ configuration, with AUPRC approximately $0.7797$. Thus, CIC-IDS2017 already shows that the most useful macrostate family can differ by scale even when the overall macrostate advantage persists. The best macrostate stage is the full macrostate, with AUPRC approximately $0.7917$, providing a modest but real gain. At $30\,$s, the same directional effect remains: entropy-vector logistic regression is again the strongest baseline, while the best macrostate stage outperforms it. Importantly, however, the strongest macrostate at $30\,$s is not the full state but the intermediate $A{+}H{+}V$ configuration, with AUPRC approximately $0.7797$.

This difference between $30\,$s and $60\,$s on CIC is important. It suggests that scale sensitivity in this paper is not merely about whether performance changes with window size, but also about which observable families are most useful at different scales. On CIC, the richest macrostate is strongest at $60\,$s, while a leaner activity--disorder--volatility representation is strongest at $30\,$s.

\begin{figure}[t]
    \centering
    \includegraphics[width=\linewidth]{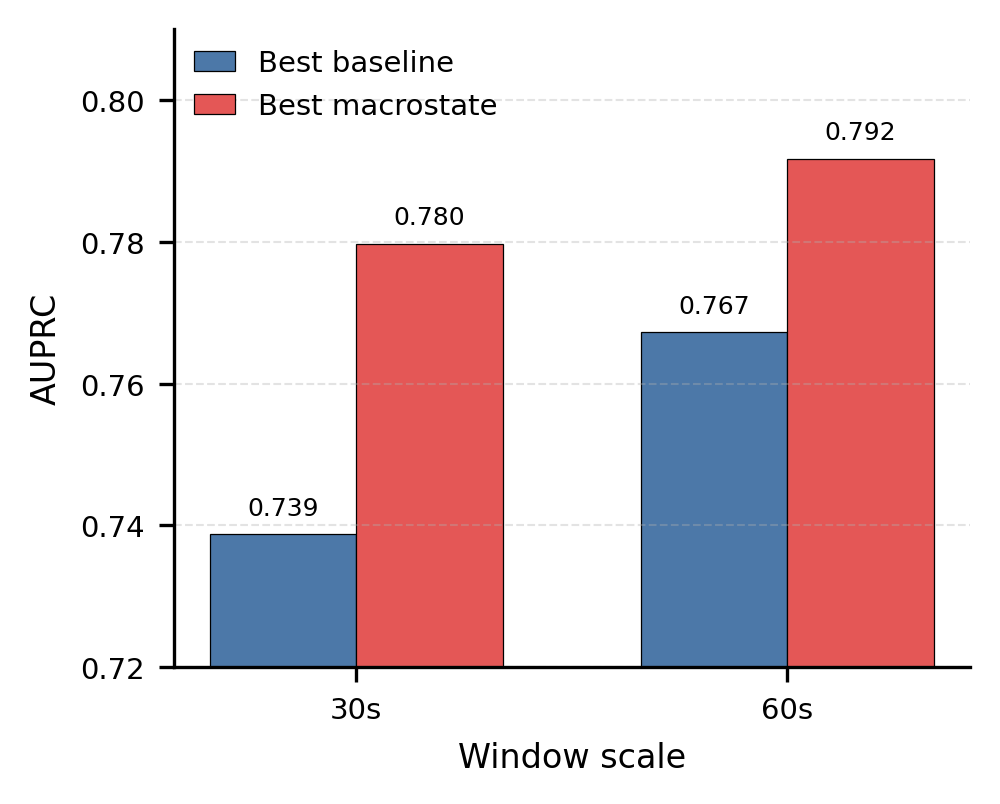}
    \caption{First-pass CIC-IDS2017 comparison at 30s and 60s using AUPRC. At both scales, the best macrostate model outperforms the best entropy-based baseline, while 60s is stronger overall. The differing best ablation stages across the two scales also indicate early scale sensitivity on CIC.}
    \label{fig:cic_30s_vs_60s}
\end{figure}

Figure~\ref{fig:cic_30s_vs_60s} shows that the CIC-IDS2017 results follow the same broad direction as UNSW while also revealing a scale-dependent difference in the strongest macrostate composition. At both scales, the best macrostate outperforms the best entropy-only baseline. However, the strongest macrostate at $30\,$s is the intermediate $A{+}H{+}V$ configuration, whereas at $60\,$s the full macrostate is best. This is useful evidence that the value of different observable families can vary with temporal scale.

\subsection{Cross-Dataset Interpretation}

Taken together, the UNSW and CIC results support a coherent empirical narrative. On both datasets, entropy-only methods remain meaningful but are not best once richer macrostate families are introduced. On both datasets, richer behavioral representations improve anomaly discrimination over entropy-only baselines. On both datasets, scale matters. And on both datasets, the most useful explanatory story is not simply one of higher or lower entropy, but one of macrostate composition: activity contextualizes disorder, structure helps identify organized reconfiguration, and the best observable family combination can change with temporal resolution.

At the same time, the results also identify clear boundaries. Transition-aware scoring is not yet a demonstrated advantage. Multiscale fusion is useful but not yet superior to the best single-scale macrostate model. And the strongest clean empirical anchor remains UNSW at $60\,$s. These limits are not weaknesses of the whole approach; rather, they define the correct boundary of this paper.

\subsection{Synthesis}

The empirical results therefore support the central claim of the paper in a disciplined but non-exaggerated form. Entropy remains important, but entropy-only formulations are insufficient as complete behavioral descriptions. Finite macrostates built from activity, disorder, structure, volatility, persistence, and deviation provide a more informative and more interpretable representation of network behavior. The gains are strongest at the $60\,$s scale on UNSW and replicate directionally at $30\,$s and $60\,$s on CIC-IDS2017. These results are therefore best read as the first empirical validation of a broader cyber-dynamics program: anomaly detection in cyber systems is not only a problem of scalar irregularity, but also a problem of behavioral state-space construction and analysis.

\section{Discussion}

The empirical results support the main conceptual claim of the paper in a careful but meaningful way. At the most practical level, the proposed macrostate representation improves anomaly discrimination over entropy-only baselines on the primary benchmark and reproduces the same directional advantage on a second benchmark. At the conceptual level, the experiments show that cyber anomaly detection is better understood as a problem of behavioral state description than as a problem of scalar disorder alone. This section interprets those findings, clarifies their limits, and situates this paper within the broader cyber dynamics program.

\subsection{Why Scalar Entropy is Not Enough}

The first empirical conclusion is that scalar entropy remains useful but is not sufficient. On both UNSW-NB15 and CIC-IDS2017, entropy-only methods provide nontrivial baseline performance, which is entirely consistent with the long-standing information-theoretic anomaly literature. However, once richer macrostate families are introduced, entropy-only formulations are no longer the strongest representation. This is exactly the empirical pattern one would expect if entropy captures only one dimension of cyber behavior while leaving structural and dynamical organization unresolved.

The ablation ladders help sharpen this conclusion. On UNSW, Figure \ref{fig:unsw_ablation_by_scale}, the move from $H$ alone to $A{+}H$, then to $A{+}H{+}S{+}V$, and finally to the full macrostate shows that the gains are not random. They follow a structured progression in which entropy becomes more useful once embedded in a broader behavioral context. On CIC, the same directional improvement appears, but with a scale-dependent difference in the strongest ablation stage. This reinforces the paper's main point: entropy is one macro-observable family, not the whole macrostate.

\subsection{Benign Drift, Structured Reorganization, and Scale}

A second key lesson is that temporal scale matters not only for performance magnitude, but for the very composition of the most informative behavioral representation. On UNSW, as shown in Fig.~\ref{fig:unsw_ablation_by_scale}, the move from $H$ alone to $A{+}H$, then to $A{+}H{+}S{+}V$, and finally to the full macrostate shows that the gains are not random. On CIC, the $60\,$s scale again supports the full macrostate, whereas the $30\,$s scale favors the leaner $A{+}H{+}V$ configuration. This is an important result because it shows that the macrostate framework is not merely learning a fixed feature template. Rather, the relative importance of activity, structure, volatility, persistence, and deviation can shift with the temporal scale at which behavior is observed.

This finding is closely aligned with the multiscale arguments developed in the theory section. Some anomalies become visible as short-horizon irregularity, others as intermediate-horizon structural reorganization, and still others only as longer regime drift. This work does not yet establish an optimal multiscale fusion strategy, but it does provide empirical support for the stronger and more basic claim that anomaly visibility is scale-dependent.

\subsection{What the Transition Results Mean}

The transition-aware results should be interpreted conservatively. The current experiments do not show a strong gain from transition-only or fused state-transition scoring. On the informative scales, state-only scoring remains best. This means that this paper should not claim a demonstrated transition benefit. Instead, the correct conclusion is narrower: the finite macrostate representation itself is already useful, while the specific transition models used in this first paper are too simple to unlock the full dynamical benefit of the framework.

This negative result is still informative. It suggests that the value of this paper lies primarily in establishing a reusable behavioral state-space, while the stronger exploitation of trajectories, transition concentration, regime escape, and recovery will belong more naturally in future work. Put differently, the state-space is working before the transition layer is mature. That is a positive architectural result even if it is not yet a positive transition result.

\subsection{Primary Strength of the Paper}

The strongest empirical contribution of this paper is therefore not that every part of the cyber dynamics program is already fully realized. It is that a finite, interpretable macrostate representation can outperform entropy-only baselines and do so in a way that is behaviorally meaningful. The gains are strongest when the representation incorporates not only disorder but also activity and structure, and they remain visible across two datasets. This is enough to justify the macrostate program as a serious and empirically grounded direction rather than merely a conceptual reframing.

\subsection{Limitations}

Several limitations remain and should be stated clearly.

First, the cleanest evidence currently comes from UNSW-NB15, especially at the $60\,$s scale. The CIC-IDS2017 results are encouraging and directionally consistent, but they are still bounded first-pass experiments rather than a full second benchmark campaign.

Second, the larger-window UNSW settings remain more imbalanced than the smaller scales even after split repair. They are usable as secondary evidence, but they should not bear the full interpretive weight of the paper.

Third, the CIC experiments were intentionally optimized for runtime, which required a fast approximation for the deviation family and a bounded two-scale evaluation rather than a full all-scale sweep. These are reasonable engineering choices for this paper, but they mean that the cross-dataset comparison is not yet perfectly symmetric.

Fourth, transition-aware scoring is not yet mature enough to support a positive empirical claim. This should be treated as a deferred strength of the broader program rather than as a demonstrated result of this paper.

Fifth, this paper remains network-centered. Although the framework is defined in a CSTS-compatible way and is intended to extend naturally to heterogeneous telemetry, the present experiments still concern network behavior rather than full endpoint--identity--provenance fusion.

\subsection{Future work}

These limitations also define the next layer of the program. Future work can take the state-space established here and focus on the genuinely dynamical questions that this paper only begins to touch: stability, regime persistence, recovery, metastability, and escape under perturbation. We can then generalize the framework from network telemetry to heterogeneous cybersecurity data, using CSTS as the canonical substrate on which richer cross-modal macrostate families are defined.

In that sense, this paper shows that finite macrostates are empirically useful, interpretable, and portable enough to justify further investment. The fact that transition scoring remains weak and multiscale fusion remains only partially successful should be seen less as failure than as a clean definition of the next frontier.

\section{Conclusion}

This paper introduced the first empirical step of a broader cyber dynamics program: a finite macrostate framework for behavioral anomaly detection in network telemetry. The central idea is that raw cyber events should be coarse-grained into a finite-dimensional macrostate whose coordinates capture activity, distributional disorder, structural organization, temporal volatility, persistence, and deviation from benign behavior. Once this representation is defined, anomaly detection can be posed not only as a question of unusual observations, but as a question of behavioral state, regime structure, and state-space evolution.

The empirical results support this shift in a disciplined way. On UNSW-NB15, the clearest evidence appears at the $60\,$s scale, where the full macrostate outperforms entropy-only baselines and structural variables contribute consistent gains beyond activity and disorder alone. On CIC-IDS2017, bounded first-pass experiments at $30\,$s and $60\,$s reproduce the same directional macrostate advantage over entropy-only baselines, while also showing that the most useful macrostate composition can vary with scale. Taken together, these results support three core claims of this paper: entropy-only summaries are insufficient, richer finite macrostates improve anomaly discrimination, and temporal scale materially affects what kinds of behavior are most visible.

At the same time, the paper also identifies clear limits. Transition-aware scoring is not yet a demonstrated advantage in its current simple form, and multiscale fusion, while useful, does not yet exceed the best single-scale macrostate model. These are not failures of the overall framework. Rather, they clarify the boundary of what this paper has established: the finite macrostate representation itself is empirically meaningful and interpretably structured, while the richer dynamical layers of the program remain open for the next papers.

More broadly, the results suggest that cyber anomaly detection should be treated less as a narrow thresholding problem over scalar irregularity and more as a problem of behavioral state-space construction over canonicalized telemetry. In that sense, this paper does not claim a completed cyber thermodynamics. It establishes a reusable formal and empirical core: microstates, windowed ensembles, finite macrostates, behavioral-state interpretation, and the first evidence that such a representation improves over entropy-only baselines across benchmark settings. That is the proper foundation on which later work in stability, recovery, heterogeneous telemetry, and broader behavioral cyber AI can be built.
\appendix

\section{Symbolic Dictionary}
\begin{table}[H]
\centering
\footnotesize
\renewcommand{\arraystretch}{1.08}
\setlength{\tabcolsep}{3pt}
\caption{Symbolic dictionary for the macrostate framework.}
\label{tab:symbolic_dictionary}
\begin{tabularx}{\columnwidth}{>{\raggedright\arraybackslash}p{0.25\columnwidth} >{\raggedright\arraybackslash}X}
\toprule
Symbol & Meaning \\
\midrule
$e$ & network microstate / event \\
$\tau(e)$ & timestamp of event $e$ \\
$W_t$ & time window $[t,t+\Delta)$ \\
$\cE_t$ & micro-ensemble induced by $W_t$ \\
$\PhiMacro$ & macrostate map \\
$X_t$ & macrostate at window $W_t$ \\
$A_t$ & activity/load coordinates \\
$H_t$ & disorder/entropy coordinates \\
$S_t$ & structural-order coordinates \\
$V_t$ & volatility coordinates \\
$P_t$ & persistence/memory coordinates \\
$C_t$ & coupling/influence coordinates \\
$D_t$ & deviation-to-baseline coordinates \\
$\cB$ & benign regime set \\
$\cA$ & anomalous region of macrostate space \\
$\cR$ & behavioral regime \\
$T$ & nominal transition operator \\
$\Delta$ & window length \\
$\delta$ & window stride \\
$\mathcal{S}_t^{\mathrm{state}}$ & state anomaly score \\
$\mathcal{S}_t^{\mathrm{trans}}$ & transition anomaly score \\
$\mathcal{S}_t^{\mathrm{fused}}$ & fused anomaly score \\
$\tau_{\mathrm{rec}}(t)$ & recovery time after perturbation at time $t$ \\
\bottomrule
\end{tabularx}
\end{table}

%
\bibliographystyle{IEEEtran}
\bibliography{ref}


\end{document}